\documentclass[reqno,11pt]{amsart}

\IfFileExists{mymtpro2.sty}{%
\usepackage[subscriptcorrection]{mymtpro2}}%
{%
}

\usepackage{a4,amssymb}

\newtheorem{theorem}{Theorem}
\newtheorem{proposition}{Proposition}

\newtheorem{lemma}{Lemma}

\theoremstyle{remark}

\newenvironment{acknowledgement}{\par\medskip\noindent\emph{Acknowledgment.}}

\newcommand{\beq}{\begin{equation}}
\newcommand{\eeq}{\end{equation}}
\newcommand{\ba}{\begin{array}}
\newcommand{\ea}{\end{array}}
\newcommand{\bea}{\begin{eqnarray}}
\newcommand{\eea}{\end{eqnarray}}

\newcommand{\R}{\mathbb{R}}

\newcommand{\Sp}{\mathbb{S}}

\newcommand{\N}{\mathbb{N}}

\newcommand{\C}{\mathbb{C}}

\begin{document}

\title[Resonance Clusters for the Hydrogen Stark Hamiltonian]{Semiclassical Szeg\"o Limit of Resonance Clusters for the Hydrogen Atom Stark Hamiltonian}

\author[P.\ D.\ Hislop]{Peter D.\ Hislop}
\address{Department of Mathematics,
    University of Kentucky,
    Lexington, Kentucky  40506-0027, USA}
\email{hislop@ms.uky.edu}

\author[C.\ Villegas-Blas]{Carlos Villegas-Blas}
\address{Universidad Nacional Aut\'onoma de M\'exico, Instituto de Matem\'aticas,
Unidad Cuernavaca, Mexico} \email{villegas@matcuer.unam.mx}


\vspace{.1in}

\begin{abstract}
We study the weighted averages of resonance
clusters for the hydrogen atom with a Stark electric field in the weak field limit.
We prove a semiclassical Szeg\"o-type theorem for resonance clusters showing
that the limiting distribution of the resonance shifts concentrates on the classical energy
surface corresponding to a rescaled eigenvalue of the hydrogen atom
Hamiltonian. This result extends Szeg\"o-type results on eigenvalue
clusters to resonance clusters. There are two new features in this work: first, the study of resonance clusters
requires the use of non self-adjoint operators, and second, the
Stark perturbation is unbounded so control of the perturbation is achieved
using localization properties of coherent states corresponding to hydrogen atom eigenvalues.
\end{abstract}

\maketitle
\thispagestyle{empty}


\tableofcontents


\section{Introduction: Semiclassical Szeg\"o limits}\label{sec:szego-intro}

The behavior of eigenvalue clusters resulting from the perturbation
of highly degenerate eigenvalues of elliptic operators on compact
manifolds has been studied by many researchers, notably by V.\
Guillemin \cite{guillemin1} and by A.\ Weinstein \cite{weinstein}.
The basic idea is the following. Suppose that $E_N$ is an eigenvalue
of an elliptic self-adjoint operator $P$ with multiplicity $d_N$,
growing with $N$. For bounded perturbations $Q_N$, satisfying $\|
Q_N \| \rightarrow 0$ as $N \rightarrow \infty$, there is a cluster
of nearby eigenvalues $E_{N,j}$, with the same total multiplicity
$d_N$ as $E_N$, that tend to $E_N$ as $N \rightarrow \infty$. The
eigenvalue shifts $\nu_{N,j}$ are defined by $\nu_{N, j} \equiv E_N
- E_{N,j}$. The basic question concerns the distribution of these
eigenvalue shifts as $N \rightarrow \infty$. Since the eigenvalue
$E_N$ is increasing with $N$, the Hamiltonian is rescaled so that
the eigenvalue $\tilde{E}$ is independent of $N$. This rescaling
results in a rescaled perturbation and rescaled eigenvalue shifts
$\tilde{\nu}_{N,j}$. If the rescaled perturbation $\tilde{Q}_N$
vanishes with a rate $\kappa (N)$, then the rescaled eigenvalue
shifts $\tilde{\nu}_{N,j}$ vanish at the same rate. Consequently,
the point measure $(1/d_N) \sum_{j=1}^{d_N} \delta ( \lambda -
\tilde{\nu}_{N,j}/ \kappa (N)) d \lambda$ should have a weak limit
as $N \rightarrow \infty$. In particular, if $\rho \in C_0(\R)$,
then, roughly speaking, one proves that \beq\label{eq:szego-intro1}
\lim_{N \rightarrow \infty} \frac{1}{d_N} \sum_{j=1}^{d_N} \rho
\left( \frac{\tilde{\nu}_{N,j}}{\kappa (N)} \right) =
\int_{\mathcal{A}} \rho (\tilde{Q} (\alpha)) ~d \mu (\alpha) , \eeq
where $\mathcal{A}$ is a parameterization of the classical
Hamiltonian orbits with energy $\tilde{E}$ and $\mu$ is an invariant measure on this energy surface.
The effective potential
$\tilde{Q}$ is the average of a re-scaled perturbation over one of
these orbits.

In the semiclassical context,
we interpret $h = 1/N$ as Planck's constant, and then this formula
\eqref{eq:szego-intro1} is what we mean by a {\it semiclassical Szeg\"o
limit} for the appropriately rescaled eigenvalue shifts of perturbed operator.
Formula \eqref{eq:szego-intro1} expresses the weak limit of the distribution function
of the eigenvalue shifts in terms of averages of the perturbation over corresponding classical orbits.

The behavior of eigenvalue clusters for bounded perturbations $V$ of the Laplacian $- \Delta_{\Sp^n}$
on $L^2 (\Sp^n)$ were studied by Guillemin
\cite{guillemin1}.
Weinstein \cite{weinstein} studied these Szeg\"o-type limits for the Laplacian on a compact
manifold perturbed by a bounded, real-valued function $V$.
In both cases, the semiclassical parameter is the index of the unperturbed eigenvalue. The integral on the right in \eqref{eq:szego-intro1} is the average of
potential perturbation $V$ over closed geodesics of the
manifold. Brummelhuis and Uribe \cite{B-U} extended these results to the study of the semiclassical Schr\"odinger
operator $- h^2 \Delta + V$ on $L^2 (\R^n)$. The potential $V \geq 0$ is smooth with $V_\infty \equiv
\liminf_{|x| \rightarrow \infty} V(x) > 0$. They studied the semiclassical behavior of the eigenvalue cluster near an energy $0 < E^2 < V_\infty$. They proved an asymptotic expansion of $Tr \rho[ (H_h^{1/2} - E) h^{-1}]$ as $h \rightarrow 0$ and related the coefficients to the classical flow for $p^2 +V$ on the energy surface $E$.

Uribe and Villegas-Blas \cite{uribe-villegas1}
extended these results by considering perturbations of the hydrogen atom Hamiltonian
by operators of the form $\epsilon (h) Q_h$ where
$Q_h$ is a zero-order pseudo-differential operator uniformly bounded in $h$
and $\epsilon (h) = \mathcal{O}(h^{1+\delta})$, for $\delta >0$.
The main novelty comes from the fact that for a fixed negative energy,
there are two types of classical orbits for the Hamiltonian flow on an energy surface for negative energy.
There are bounded periodic orbits corresponding to nonzero angular momentum,
and there are unbounded collision orbits with zero angular momentum.
Uribe-Villegas \cite{uribe-villegas1} used Moser's regularization of collision orbits
so that all orbits are considered periodic orbits. In this regularization, all orbits
correspond to geodesics on the sphere $\Sp^3$. Those passing through the north pole are the collision orbits.
The geodesics on $\Sp^3$ are parameterized by a certain five-dimensional set $\mathcal{A}$
described in Appendix 1, section \ref{sec:app1-coherentstates1}.

In this paper, we extend these results to resonances of the Stark hydrogen Hamiltonian. We prove a
Szeg\"o-type result on the semiclassical behavior of the distribution of
the resonance shifts. To explain this in more detail,
let $E_N (h) =- 1 / (2 h^2 N^2)$ be an eigenvalue of the hydrogen atom Hamiltonian
$H_V(h) = - (1/2) h^2 \Delta - |x|^{-1}$, defined on $L^2 (\R^3)$ (see (\ref{eq:hydro1})),
with multiplicity $d_N = N^2$. Applying an external electric field of strength $F > 0$,
the resulting Hamiltonian $H_V(h,F) = H_V(h) + F \epsilon (h) x_1$, called here the {\it Stark hydrogen atom Hamiltonian} (see (\ref{eq:stark1})),
has purely absolutely continuous spectrum equal to the
real line $\R$. We will keep $F > 0$ constant
assume that $\epsilon (h)$ vanishes as $h \rightarrow 0$ corresponding to weak field limit.
Under the perturbation by the electric field, the
eigenvalue $E_N(h) <0$ gives rise to a cluster of nearby resonances $z_{N,i}(h, F),
i= 1, \ldots , K_N$, with total algebraic multiplicity equal to $d_N$. We have $\Re z_{N,i} (h, F) \sim E_N (h)$
and the imaginary part of the resonance $\Im  z_{N,i}(h,F)$ is exponentially small
in $1/(hF)$.

To study the semiclassical limit, we take $h = 1/N$ as in Uribe-Villegas \cite{uribe-villegas1}. Then, the family of
hydrogen atom Hamiltonians
$H_V ( 1/N, F=0)$ has a fixed eigenvalue $E_N (1/N)=  -1 /2$. The Stark hydrogen atom
Hamiltonian $H_V (1/N, F)$ has a cluster of nearby resonances $z_{N,i} ( 1/N, F)$ that
converge to $-1/2$ as $N \rightarrow \infty$. Our main result is the following Szeg\"o-type theorem for this resonance cluster in the large $N$ limit corresponding to a weak electric field.

\begin{theorem}\label{thm:main1}
Let $F > 0$ be fixed, and let
$\rho$ be a function analytic in a disk of radius $3 F$ about $z = 0$. Let
$\epsilon (h) = h^{6 + \delta}$, for some $\delta > 0$, small, and take $h = 1/N$, with $N \in \N$.
For the resonance cluster $\{ z_{N,i}(1/N,F)\}$ near $E_N( 1/N ) = -1/2$,
we have
\bea\label{eq:szego1}
\lefteqn{\lim_{N \rightarrow \infty} \frac{1}{d_N}
\sum_{j=1}^{d_N} \rho \left( \frac{ z_{N,i} (1/N, F) - E_N (1/N) }{ \epsilon
(1/N) } \right)}  \nonumber \\
  &=&  \int_{\Sigma (-1/2) } \rho \left( \frac{1}{2 \pi} \int_0^{2 \pi} F
\cdot ( \tilde{\phi}_t ( x, p ) )_1 ~dt \right) ~ d\mu_L (x, p)
,
 \eea
where $\tilde{\phi}_t$ is the Hamiltonian flow for the Kepler problem
on the energy surface $\Sigma (-1/2)$ with collision orbits treated as in \cite{uribe-villegas1},
and $\tilde{\phi}_t ( x(t), p(t) ) )_1$ is the projection of this flow onto the first coordinate axis
$x_1$. The measure $\mu_L$ is the normalized
Liouville measure on restricted to the energy surface $\Sigma (-1/2)$.
\end{theorem}

This result parallels and extends the result of Uribe and Villegas-Blas \cite{uribe-villegas1} on
eigenvalue clusters formed by bounded perturbations $Q_h$ of the hydrogen atom Hamiltonian.
There are three new main components in this work. The first is that the Stark perturbation is unbounded.
The bounded perturbation $Q_h$ of Uribe-Villegas \cite{uribe-villegas1} is replaced by the unbounded perturbation
$F \epsilon (h) x_1$ with small field strength as $h \rightarrow 0$.
Control of the unbounded perturbation is obtained by utilizing the localization properties of coherent states
of the hydrogen atom Hamiltonian.
The second is the fact that we work with resonances that appear as eigenvalues of non self-adjoint operators.
Consequently, many estimates appearing in \cite{uribe-villegas1} have to be established for non self-adjoint
operators.
Thirdly, we use a semiclassical result of Thomas-Villegas \cite[Theorem 4.2]{thomas-villegas1}
to evaluate the trace of the Stark perturbation restricted to certain finite-dimensional subspaces
(see Theorem \ref{th:szego-limit1}.)

\subsection{Contents}\label{subsec:contents1}

In section \ref{sec:scale1}, we rescale the Stark hydrogen atom Hamiltonian using the dilation group. This establishes a countable family of
rescaled Hamiltonians all having  a fixed eigenvalue $-1/2$. In section \ref{sec:resonance1}, we review the results of Herbst
\cite{herbst1} on resonances for the Stark hydrogen atom Hamiltonian. We prove several important resolvent estimates necessary for
our work, extending some estimates of Herbst \cite{herbst1}. The main semiclassical result is proved in section \ref{sec:trace1}. We show that the semiclassical Szeg\"o-type limit can be obtained by evaluating the trace
of the Stark perturbation restricted to the eigenspace of the hydrogen atom Hamiltonian.
This requires decay properties of the analytically
continued coherent states. The final part of the proof of Theorem \ref{thm:main1} is proved in section \ref{sec:poly1}. We apply a
theorem of Thomas-Villegas
\cite{thomas-villegas1} to polynomially bounded perturbations in order to evaluate the large $N$ limit of the trace of the Stark perturbation
restricted to the hydrogen atom Hamiltonian eigenspace.

\vspace{.1in}


\begin{acknowledgement}

PDH was partially supported by NSF grant 0803379 during the time this work was done.
CV-B was partially supported by the project PAPIIT-UNAM IN 109610-2 and thanks
the members of the Department of Mathematics of the University of Kentucky for their hospitality during a visit.

\end{acknowledgement}


\section{Scaling}\label{sec:scale1}


The hydrogen atom Hamiltonian $H_V (h)$ with the semiclassical parameter $h$ acts on the Hilbert space
$L^2 (\R^3)$. It is the self-adjoint operator given by
\beq\label{eq:hydro1}
H_V(h) = - \frac{h^2}{2} \Delta - \frac{1}{|x|}.
\eeq
The discrete spectrum consists of an infinite family of eigenvalues $E_k (h)$
\beq\label{eq:ev1}
E_k(h) = \frac{-1}{2 h^2 k^2}, ~~k \in \N,
\eeq
each eigenvalue having multiplicity $k^2$.
The spacing between successive eigenvalues is $\mathcal{O} (k^{-3})$.

With the choice of  $h = 1/N$ and $k = N$, we see that $E_{k=N}(h=1/N) = -1/2$ is in the spectra of the
countable family of Hamiltonians $H_V(1/N)$, $N \in \N$. The multiplicity of the eigenvalue
$-1/2$ grows as $N^2$.

The unscaled Stark hydrogen Hamiltonian is
\bea\label{eq:stark1}
H_V(h,F) &=& - \frac{h^2}{2} \Delta - \frac{1}{|x|} +  \epsilon ( h) F  x_1 \nonumber \\
 &=& H_V(h) + w_h (F) ,
\eea where $F \geq 0$ is the electric field strength and we have
chosen the $x_1$-direction for the field.
We consider a parameter $\epsilon (h) = h^{K+ \delta}$, for $0 < \delta < 1$. We will choose $K \in \N$ below.

The dilation group $D_\alpha$, $\alpha > 0$, a representation of the multiplicative group $\R^+$,
has a unitary implementation of $L^2 ( \R^d)$ given by
\beq\label{eq:dilation0}
(D_\alpha f)(x) = \alpha^{d/2} f(\alpha x).
\eeq
We scale the Hamiltonian in (\ref{eq:stark1}) by $\alpha = h^2$:
\beq\label{eq:scale2}
D_{h^2} H_V(h, F) D_{h^{-2}} = \frac{1}{h^2}
\left( - \frac{1}{2} \Delta - \frac{1}{|x|} + h^4 \epsilon (h) F x_1
\right) .
\eeq
We call this rescaled Hamiltonian $S_h (F)$ so that
\bea\label{eq:rescale1}
S_h(F) & = & - \frac{1}{2} \Delta - \frac{1}{|x|} + h^4 \epsilon (h) F x_1 \nonumber \\
 & = & H_V  + h^4 \epsilon (h) F x_1,
\eea
where we write $H_V \equiv H_V(1)$ and $W_h (F) = h^4 \epsilon (h) F x_1$ is the rescaled perturbation.
Note that for $F=0$, the eigenvalues of $S_h(0)$ are given by $E_k(1) = - 1/ (2 k^2)$,
with $k \in \N$. We will keep $F > 0$ fixed. The effective electric field is $h^4 \epsilon (h) F$
and we can make this small be taking $h \rightarrow 0$, or, equivalently, with $h = 1 / N$,
by taking $N \rightarrow \infty$.


\section{Resonances of the Stark hydrogen Hamiltonian}\label{sec:resonance1}

The hydrogen atom Hamiltonian $H_V (h, F)$ with an external electric
field $F x_1$ is given in (\ref{eq:stark1}). We write $H_0(h,F)$ for the Stark
Hamiltonian with $V=0$. {\it For this section
only, we take $h = 1$ and write $H_V (F) = H_V(1,F)$, for the hydrogen atom Stark Hamiltonian,
and $H_V = H_V(1,0)$, when the field $F=0$.
The Stark Hamiltonian with
$V=0$ and $h = 1$ is denoted by $H_0(F) = H_0 (1,F) = - (1/2)\Delta + \epsilon (h) F x_1$.
and $H_V = H_V(1,0)$.} For $F \neq 0$, the spectrum of $H_V(F)$ is
purely absolutely continuous and equal to $\R$. We are interested in
the fate of the negative eigenvalues of $H_V$ when the field
$F$ is turned on. We review the results of Herbst \cite{herbst1} on
the resonances associated with the $F \neq 0$ case. Herbst's article
deals with more general Stark Hamiltonians but we cite and use his results
only for the hydrogen atom case of interest here for which $V$ is the Coulomb potential.


\subsection{Dilation analyticity}\label{subsec:dilanalyt1}

The dilated Stark hydrogen Hamiltonian is obtained by conjugating  (\ref{eq:stark1}) with the unitary
dilation group $D_\alpha$ defined in (\ref{eq:dilation0}). We take
$\alpha = e^\theta > 0$, for $\theta \in \R$. We obtain
\bea\label{eq:dilation1}
H_V( F, \theta) &=& D_{exp(\theta)} H_V(F) D_{exp(- \theta)} \nonumber \\
 &=&   - \frac{e^{-2 \theta} }{2} \Delta - \frac{e^{-\theta} }{|x|} + e^\theta  F  x_1 .
\eea

We are interested in extending this formula to $\theta$ with $\Im \theta \neq 0$.
There are two properties that need to be checked: the analyticity of the potential, and the analyticity of the Stark
Hamiltonian. Since the potential is a Coulomb potential, we have $V(\theta) = e^{-\theta} V$,
so $V(\theta)$ is a type A analytic family of operators. Furthermore, as $V(-\Delta + 1)^{-1}$ is a compact operator, it
follows that $V(\theta) ( -\Delta + 1)^{-1}$ is a compact operator-valued analytic function of $\theta$ for $\theta \in \C$.
In accordance with the Herbst's hypothesis \cite[p.\ 287]{herbst1}, we may take $\theta_0 = \pi / 3$,
the maximum width of the strip of analyticity allowed by the purely Stark Hamiltonian $H_0(F) = - (1/2) \Delta + Fx_1$.
In order to understand the origin of the bound $\pi / 3$,
we note that for $F \in \R$
\beq\label{eq:pi3-1}
H_0 (F, \theta) = - (1/2) e^{-2 \theta} \Delta + F e^\theta x_1 = e^{-2 \theta}[ - (1/2) \Delta + F e^{3 \theta} x_1 ] .
\eeq
Herbst proved results on the operator $- (1/2) \Delta + F e^{3 \theta} x_1 $ when $\theta$ becomes complex in
section II of \cite{herbst1}. The effective electric field $F e^{3 \theta}$ has a nonzero imaginary part,
necessary for Herbst's results, only if $0 < | \Im \theta| < \pi / 3$.

Herbst \cite{herbst1}
proved the following theorem about the dilated
operator $H_V(F,\theta)$. We consider $F > 0$ fixed.

We recall that for a closed operator with an isolated eigenvalue $z_0$, the algebraic multiplicity
of the eigenvalue is defined as the dimension of the range of the corresponding Riesz projector.

\begin{theorem}\cite[Theorem III.2]{herbst1}\label{th:analyticity1}
For $0 < \Im \theta <  \pi /3$,
the operator $H_V(F, \theta)$ is closed on $D(- \Delta) \cap D( M_{x_1})$.
The operator family $H_V( F, \theta)$ is an analytic family of type A operators in $\theta$.
The spectrum on $H_V( F, \theta)$ is discrete, independent of $\theta$, and the algebraic
multiplicity of each eigenvalue is independent of $\theta$.
\end{theorem}


\subsection{Resonances}\label{subsec:resonances1}

Herbst \cite{herbst1} proved that for $V=0$,
the closed operator $H_0 ( F , \theta)$, $F \neq 0$, has no spectrum for $0 < | \Im \theta| < \pi / 3$.
As stated in Theorem \ref{th:analyticity1}, Herbst also showed, using the techniques of dilation analyticity,
that for $0 < \Im \theta <  \pi / 3$, the non self-adjoint Hamiltonian $H_V(F,\theta)$, with $V$ a Coulomb potential,
has isolated eigenvalues with finite algebraic multiplicity.
Furthermore, Herbst proved that these eigenvalues are connected to the
eigenvalues of the $F = 0$ and $\Im \theta = 0$ operators.

\begin{theorem}\cite[Theorem III.3]{herbst1}\label{th:reson1}
Suppose that $E_0$ is a negative eigenvalue of $H_V$, defined in (\ref{eq:hydro1}) with $h=1$,
of multiplicity $N_0$.
Then for $F > 0$ small, there are exactly $N_0$ eigenvalues, counting algebraic
multiplicity, of $H_V (F, \theta)$, as defined in (\ref{eq:dilation1})
with $0 < \Im \theta <  \pi / 3$, nearby,
and as $F \rightarrow 0^+$, these converge to $E_0$.
\end{theorem}

We also apply Theorems \ref{th:analyticity1} and \ref{th:reson1}
to the scaled operator $S_h(F)$ defined in (\ref{eq:rescale1})
with $F$ of the theorem replaced by $h^4 \epsilon (h) F$, and take $h = 1/N$.
For any fixed $N \in \N$, we consider the resonance cluster $\{ z_{N,j} (h, F) \}$
of $S_h(F, \theta)$ near the eigenvalue $-1 / (2 N^2)$ of $H_V( \theta)$.
Note that the operator $S_h(F)$ has an effective electric field
with strength $h^4 \epsilon (h) F$ that vanishes as $h \rightarrow 0$.
Hence, Theorem \ref{th:reson1} states that the resonances $z_{N,j} (h, F)$
converge to the $N^2$-degenerate eigenvalue $E_N$ as $h \rightarrow 0$.


\subsection{Resolvent estimates}\label{subsec:eresolvent1}

We summarize the resolvent estimates needed from \cite{herbst1}.
We recall that for a closed operator $A$ with domain $D(A)$, the {\it numerical range} of $A$, denoted $W(A)$,
is the smallest convex set generated by $\{ (u, Au) ~|~ u \in D(A) \}$.
We let $H_0(F) = -(1/2) \Delta + Fx_1$ be the Stark Hamiltonian.
Following Herbst \cite{herbst1}, we review the results on Stark Hamiltonians with complex electric fields.

\begin{proposition}\cite[Theorem II.1]{herbst1}
\label{prop:stark1}
We write $F = E e^{i \phi}$, with $E , \phi \in \R$, $E \neq 0$, and $0 < |\phi| < \pi / 3$.
\begin{enumerate}
\item The spectrum of $H_0(F)$ is empty.
\item The numerical range of $H_0(F)$ is the half-plane
\beq\label{eq:nr-stark1}
W(H_0(F))= \left\{ z \in \C ~|~ \Re z > \left( \frac{\cos \phi }{\sin \phi } \right) \Im z \right\},
\eeq
independent of $E \neq 0$.
\item The resolvent is bounded
\beq\label{eq:stark-res1}
\| (H_0 (F) - z)^{-1} \| \leq [\mbox{dist} (z, W(H_0(F)) ) ]^{-1} .
\eeq
\end{enumerate}
\end{proposition}

We now consider the dilated Stark Hamiltonian
$H_0(1,F,\theta) \equiv H_0(F, \theta)$, as defined in (\ref{eq:dilation1}),
\beq\label{eq:dilation2}
H_0(F, \theta) = - (1/2) e^{-2 \theta} \Delta + e^\theta F x_1 .
\eeq
The following operator plays an important role in the analysis:
\beq\label{eq:kernel1}
K(F, \theta, z) \equiv V(\theta) (H_0(F,\theta) - z)^{-1},
\eeq
where the dilated Coulomb potential is given by
\beq\label{eq:dilation3}
V(\theta) = \frac{e^{-\theta}}{|x|} .
\eeq

We prove a convergence estimate for $K(F, \theta ,z) - K(0, \theta ,z)$ with a precise rate of convergence
as $F \rightarrow 0$ (recall that $h=1$ here).
This estimate is possible since the potential is a Coulomb potential.

We recall the basic resolvent estimates. Let $H_0 (\theta) = - (1/2) e^{-2 \theta} \Delta$, and
$H_0 (F, \theta) = - (1/2) e^{-2 \theta} \Delta + F e^\theta x_1$ be the Stark Hamiltonian.
For any $F  \neq 0$,
we have the following basic estimate from Proposition \ref{prop:stark1}:
\beq\label{eq:starkest1}
\|(z - H_0 (F, \theta) )^{-1} \| \leq 1 / d(z, W(H_0(F, \theta) ) .
\eeq
Let $\gamma_N$ be a simple closed contour about $\tilde{E}_N = - 1 / 2 N^2$ of radius $1 / (8 N^3)$.
For $z \in \gamma_N$, we have
\beq\label{eq:starkest2}
\|(z - H_0 (\theta) )^{-1} \| \leq 1 / d(z, e^{-2 \theta} \R^+) = \mathcal{O}(N^{2}).
\eeq
The contour $\gamma_N$ is chosen so that it contains only one eigenvalue $\tilde{E}_N$ of
$H_V( \theta)$. Recall that
$V(\theta) ( H_0 (\theta) + 1)^{-1}$ is a compact analytic operator valued function for
$| \Im \theta | < \pi / 3$.

\begin{proposition}\cite[Proposition III.1]{herbst1}\label{prop:stark2}
\begin{enumerate}
\item The operator $K(F, \theta, z)$ is compact and jointly analytic in $(z, \theta)$
on the region
\beq\label{eq:unif-set1}
\{ ( \theta, z) ~|~ z \in \C, ~~0 < | \Im \theta| < \pi / 3  \}.
\eeq
\item We have the following convergence on the contour $\gamma_N$ with $0 < | \Im \theta | < \pi / 3$:
\beq\label{eq:kernel-conv1}
\| K(F, \theta, z) - K(0, \theta, z) \| = \mathcal{O} ( F N^4 ) ,
\eeq
as $F \rightarrow 0$. This convergence is uniform on the larger
set $\{ ( \theta, z) ~|~ d(z, W(H_0(F, \theta))) > 0, ~0 < | \Im \theta| < \pi / 3  \}$.
\end{enumerate}
\end{proposition}

We need the following lemma summarizing several key estimates on resolvents.

\begin{lemma}\label{lemma:resolventest1}
Let $z \in \gamma_N$ and $0 < |\Im \theta| < \pi / 3 $.
\begin{enumerate}
\item $\|(H_0 (F, \theta) - z)^{-1} \| = \mathcal{O}(N^2)$

\item $\|(H_0 ( \theta)-z)^{-1} \| = \mathcal{O}(N^2)$

\item $\|V(\theta)(H_0 ( \theta)-z)^{-1} \| = \mathcal{O}(N)$

\item $\|e^{-2 \theta} p_1 (H_0 ( \theta)-z)^{-1} \| = \mathcal{O}(N^2)$, where $p_1 = -i \partial / \partial x_1$.

\end{enumerate}
\end{lemma}

\begin{proof}
\noindent
1. The first estimate follows from the bound
\eqref{eq:stark-res1} and the fact that the numerical range is a half-plane located a distance $\mathcal{O}(1 / N^2)$
from the contour $\gamma_N$.

\noindent
2. The second estimate follows similarly as the spectrum of $H_0 (\theta)$ is the half line $e^{-
2 \Im \theta i} \R^+$, see \eqref{eq:starkest1}.

\noindent
3. The third estimate requires the following bound. Let $C_3$ denote the constant
$C_3 = (2 \pi)^{- 3/2} ( \int_{\R^3} (1 + | p|^2)^{-2} ~d^3 p )^{1/2}$.
For any $\psi \in H^2 (\R^3)$, and for any $\lambda > 0$, we have
\beq\label{eq:sobolev1}
\| \psi \|_\infty \leq \frac{C_3}{\lambda^{1/2}} \| \Delta \psi \| + C_3 \lambda^{3/2} \| \psi \| .
\eeq
This follows from the Sobolev embedding theorem and standard estimates with the Fourier transform.
We decompose the Coulomb potential as $V = V \chi_{B_N(0)} +  V(1 - \chi_{B_N(0)}) \equiv
 V_2 + V_\infty$, where $\chi_{B_N(0)}$ is the characteristic function on the ball of radius $N > 0$ centered
at the origin. We have that $V_2 \in L^2 (\R^3)$ and $V_\infty \in L^\infty (\R^3)$,
with $\| V_2 \| = \omega_3^{1/2} N^{1/2}$ and $\| V_\infty \|_\infty = 1 / N$. With the help of
\eqref{eq:sobolev1}, and choosing $\lambda = 1/N$, we write
\bea\label{eq:sobolev2}
\| V \psi \| & \leq  & \|V_2 \| \| \psi \|_\infty + \| V_\infty \|_\infty \| \psi \| \nonumber \\
 &\leq & ( \omega_3 N)^{1/2} C_3  \| \Delta \psi \| + ( \omega_3^{1/2} C_3 + 1) N^{-1} \| \psi\|,
\eea
where $\omega_3 = 4 \pi$.
Recall that $(H_0 (\theta) - z)^{-1} : L^2 (\R^3) \rightarrow H^2 (\R^3)$ and that $|z| = \mathcal{O}(N^{-2})$
for $z \in \gamma_N$. Taking estimate \eqref{eq:sobolev2} with $\psi =
(H_0 (\theta) - z)^{-1} \phi$, for any $\phi \in L^2 (\R^3)$,
together with estimate (2), we easily obtain estimate (3).

\noindent
4. The proof of estimate (4) follows from
\beq\label{eq:p-est1}
\| e^{-2\theta} p_1 (H_0 (\theta) - z)^{-1} \| \leq \max \left\{ | p_1 ( |p|^2 - e^{2\theta} z)^{-1}
\right\} = \mathcal{O}(N^2) ,
\eeq
since $z \in \gamma_N$.

\end{proof}

We can now give the proof of Proposition \ref{prop:stark2}.

\begin{proof}
\noindent
1. As in Herbst \cite{herbst1}, we write
\beq\label{eq:stark-kernel1}
K(F, \theta, z) = V(\theta) (- \Delta +1 )^{-1} J(\theta, z),
\eeq
where
\beq\label{eq:stark-kernel2}
J(\theta, z) = (-\Delta + 1)(H_0(F, \theta) - z)^{-1} .
\eeq
The Coulomb potential has the property that $V(\theta) (- \Delta +1 )^{-1}$ is an analytic, compact operator-valued
function for any $\theta \in \C$.
The quadratic estimate \eqref{eq:QE2} implies
that $J(\theta, z)$ is bounded. For an appropriately defined circle $\gamma$, a contour integral representation
\beq\label{eq:J-contour1}
J(\theta, z) = (2 \pi i)^{-2} \int_\gamma \int_\gamma ~dw ~d\phi (w-z)^{-1} (\phi - \theta)^{-1} J( \phi, w)
\eeq
is used to verify that $J(\theta, z)$ is analytic in the region described in the proposition: $0 < | \Im \theta| < \pi /3 $
and $z \in \C$.

\noindent
2. Using the resolvent formula, we write the difference on the left in \eqref{eq:kernel-conv1} as
\bea\label{eq:resolvest1}
\lefteqn{K(F,  \theta, z)- K(0,\theta, z) } \nonumber \\
  & = &  - F e^\theta x_1 V(\theta) (H_0(\theta) -z)^{-1} (H_0(F, \theta) - z)^{-1}  \nonumber \\
 & & + F e^\theta V(\theta) (H_0(\theta) - z)^{-1} [ H_0(\theta) , x_1] (H_0(\theta) -z)^{-1} (H_0(F, \theta) - z)^{-1}. \nonumber \\
 & &
 \eea
The commutator $[ H_0(\theta) , x_1] = - 2 e^{-2\theta}i p_1$, where $p_1 = -i \partial / \partial x_1$. Note that
$\| x_1 V(\theta) \| \leq e^{- \Re \theta}$, since the potential is Coulombic. From
the resolvent estimates in Lemma \ref{lemma:resolventest1}, we obtain
\beq\label{eq:conv2}
\| K( F, \theta, z) - K( \theta, z) \|  \leq  F e^{\Re \theta} \| (H_0 (F, \theta) - z)^{-1} \| \{ A + B \} ,
\eeq
where
\beq\label{eq:conv2-1}
A = e^{-\Re \theta} \| ( H_0 (\theta) - z)^{-1} \| = \mathcal{O}(N^2),
\eeq
and
\beq\label{eq:conv2-2}
B = \| V(\theta) (H_0(\theta) -z)^{-1} \| + 2 \| e^{-2 \theta} p_1 (H_0(\theta) -z)^{-1}\| = \mathcal{O}(N^2).
\eeq
Consequently, we find from \eqref{eq:conv2}--\eqref{eq:conv2-2} that
\beq\label{eq:resolventest1-2}
 \| K(F,  \theta, z)- K(0, \theta, z) \| =  \mathcal{O}(F N^4).
 \eeq
This proves part (2) of the proposition.
\end{proof}

We recall that our $F$ is $F h^4 \epsilon (h)$ so that  with $h = 1/N$ and $\epsilon (h) = h^{K+\delta}$,
part (2) of Proposition \ref{prop:stark2} states that uniformly for $z \in \gamma_N$, with $| \gamma_N| = 2 \pi (8 N^3 )^{-1}$, we have
\beq\label{eq:kernel-conv1-2}
\| K(F, \theta, z) - K(0, \theta, z) \| = \mathcal{O} (  N^{-K - \delta} ) ,
\eeq
as $N \rightarrow \infty$.

%
%

\section{A semiclassical trace identity for resonance clusters}\label{sec:trace1}

We now return to the scaled Hamiltonian $S_h (F) = H_V + W_h (F)$,
with $F > 0$ fixed, the perturbation $W_h(F) = h^4 \epsilon (h) F
$, with $\epsilon (h) = h^{K + \delta}$, and $K \geq{6}$. We will
take $h = 1/N$ and consider $N \rightarrow \infty$. We need a
basic trace identity relating the resonance shifts $z_{N,j} (1/N,
F) - E_N (1/N)$, with $E_N(1/N) = -1/2$, to the eigenvalues of a
reduced, finite dimensional matrix obtained from the Stark
perturbation $W_h(F)$. This is the main result of this section
stated in Theorem \ref{thm:trace1}.

In section \ref{sec:resonance1}, the dilation was written as $e^{\theta}$. The real part of $\theta$ does not
affect Theorems \ref{th:analyticity1} and \ref{th:reson1}. Consequently, \emph{we will now write the dilation as
$e^{i \theta}$, with $\theta \in \R$ and in the range $0 < |\theta| < \pi / 3 $.} The operators are obtained by analytic continuation
as discussed in section \ref{sec:resonance1}.

As above, we fix $N \in \N$.
%
%
We study the non self-adjoint operator $S_h(F, \theta)$ obtained from $S_h(F)$ in (\ref{eq:rescale1}) by dilation $D_{exp(i
\theta)}$, for $\theta \in \R$ as above:
\beq\label{eq:rescale2}
S_h (F, \theta) = D_{exp(i \theta)} S_h(F)  D_{exp(- i \theta)} = H_V(\theta)
+ W_h (F, \theta),
\eeq
where the dilated, scaled hydrogen atom
Hamiltonian is
\beq\label{eq:dilhydro1}
H_V(\theta) =
D_{exp(i \theta)}H_V D_{exp(- i \theta)} = - \frac{e^{-2 i \theta}}{2}
\Delta - \frac{e^{- i \theta}}{|x|},
\eeq
and the dilated perturbation is
\beq\label{eq:pert1}
W_h (F, \theta) =  D_{exp(i \theta)} W_h(F)
D_{exp(- i \theta)} = h^4 \epsilon (h) e^{ i \theta} F x_1.
\eeq
We write
$W_h(F) = W_h(F, 0)$
and note that $W_h (F) = h^4 \epsilon (h) F x_1$ is self-adjoint.

Let $\Pi_N^{0}$ be the orthogonal
projector for the eigenvalue $\tilde{E}_N = - 1 / (2 N^2)$ of the
scaled hydrogen atom Hamiltonian $H_V$ defined in
(\ref{eq:rescale1}).
Under dilation, these remain eigenvalues of $H_V(\theta)$.
Let $P_N(\theta)$ be the projector
for the resonance cluster $\{ \tilde{z}_{N,i} (h,F) \}$ near
$\tilde{E}_N$ of the non self-adjoint operator $S_h (F, \theta)$.
We write for the resonance shift
\beq\label{eq:shift1}
\tilde{z}_{N, i}(1/N, F) = \tilde{E}_N + \nu_{N,i}, ~~\nu_{N,i} \in \C .
\eeq
The following trace estimate for the resonance shifts is the main result of this section.

\begin{theorem}\label{thm:trace1}
Let $\nu_{N,i}$ be the complex resonance shifts defined
in (\ref{eq:shift1}), and let $\tau_{N,i}$
be the eigenvalues of the self-adjoint operator $\Pi_N^0 W_h(F ) \Pi_N^0$.
For any $m \in \N$, and for $h = 1/N$, we have the following trace formula:
\beq\label{eq:trace2}
\frac{1}{d_N} \sum_{i}^{d_N} \left( \frac{ \nu_{N , i} }{ h^2 \epsilon (h) } \right)^m =
\frac{1}{d_N} \sum_{i = 1}^{d_N} \left( \frac{ \tau_{N, i} }{ h^2 \epsilon (h) } \right)^m +
\mathcal{O} \left( \frac{1}{N^\beta}
\right) ,
\eeq
for some $\beta > 0$.
\end{theorem}

The proof of Theorem \ref{thm:trace1} requires two main steps. In the first, we express the
left side of (\ref{eq:trace2}) in terms of the trace of the operator
$P_N (\theta) ( S_h (F, \theta) - \tilde{E}_N )^m P_N (\theta)$. In the second step,
we evaluate the trace of this operator and express it in terms of the finite-rank
operator $\Pi_N^0 W_h(F ) \Pi_N^0$.


\subsection{Step 1. A trace calculation.}

Since $S_h (F, \theta)$ is non self-adjoint, the projector $P_N(\theta)$ is not self-adjoint.
The range of $P_N(\theta)$ is a finite-dimensional subspace $\mathcal{E}_N$ with a dimension
$N^2$ that is equal to the geometric multiplicity of the eigenvalue $\tilde{E}_N$ of $H_V$.
Let $\tilde{z}_{N,j}(h,F)$, $j=1, \ldots, K$, with $1 \leq K \leq N^2$
be a listing of the \textbf{distinct} resonances
that converge to $\tilde{E}_N$ as $N \rightarrow \infty$. Let $P_{N,j} (\theta)$
be the projector onto the generalized eigenspace $\mathcal{E}_{N,j}$ corresponding to the
resonance $\tilde{z}_{N,j}(h,F)$. The subspace $\mathcal{E}_N$
admits a direct sum decomposition $\mathcal{E}_N = \oplus_{j=1}^K \mathcal{E}_{N,j}$,
where the finite-dimensional subspaces $\mathcal{E}_{N,j}, ~j=1, \ldots, K$ have the
following properties:
\begin{enumerate}
\item $\mathcal{E}_{N,j} = \mbox{ran} ~ P_{N,j}(\theta)$ and $\mbox{dim} \mathcal{E}_{N,j} = m_j$
\item $P_N(\theta) = \sum_{j=1}^K P_{N,j}(\theta)$
\item the algebraic multiplicity of the resonance $z_{N,j}$ is given by $m_j$
\item $N^2 = \sum_{j=1}^K m_j$
\item the projectors satisfy $P_{N,j} (\theta ) P_{N,m}(\theta)
= \delta_{jm} P_{N,j} (\theta)$
\item on the invariant subspace $\mathcal{E}_{N,j}$,
we have $S_h(F, \theta) | \mathcal{E}_{N,j} = \tilde{z}_{N,j} I_{\mathcal{E}_{N,j}} +
D_{N,j}$, where $D_{N,j}$ is nilpotent with order $m_j$ and commutes with $S_h(F , \theta)$
\item $\ker( S_h (F, \theta) - \tilde{z}_{N,j} )^{m_j} = \mathcal{E}_{N,j}$.
\end{enumerate}
We refer to Kato \cite[chapter III, section 6.5]{kato} for proofs of all these properties.

The main part of the proof of step 1 is the evaluation of the trace
\beq\label{eq:trace0}
Tr \left( P_N (\theta) ( S_h (F, \theta) - \tilde{E}_N )^m P_N (\theta) \right) .
\eeq
Using the facts listed above, we find
\bea\label{eq:trace01}
Tr \left( P_N (\theta) ( S_h (F, \theta) - \tilde{E}_N )^m P_N (\theta) \right)
 &=& \sum_{j=1}^K Tr \left( P_{N,j} (\theta) ( S_h (F, \theta) - \tilde{E}_N )^m P_{N,j} (\theta) \right)
   \nonumber \\
   & = & \sum_{j=1}^K Tr \left( P_{N,j} (\theta) ( \tilde{z}_{N,j} (h,F)
   - \tilde{E}_N + D_{N,j} )^m P_{N,j} (\theta) \right)
    \nonumber \\
 &=& \sum_{j=1}^K  ( \tilde{z}_{N,j} (h,F) - \tilde{E}_N )^m Tr P_{N,j} \nonumber \\
 &=&  \sum_{j=1}^K  m_j ( \tilde{z}_{N,j} (h,F) - \tilde{E}_N )^m  \nonumber \\
 &=& \sum_{j=1}^K  m_j ( \nu_{N,j})^m .
\eea
We also used the facts that $Tr P_{N,j} (\theta) = m_j$, and that $Tr P_{N,j} (\theta)
D_{N,j} = 0$, since $D_{N,j}$ is nilpotent and $D_{N,j}P_{N,j} = P_{N,j} D_{N,j} = D_{N,j}$.

The distinct
complex resonance shifts $\nu_{N,j}$ are defined in (\ref{eq:shift1}). We now change notation
and list the shifts with multiplicity included. That is, $\nu_{N,j}$ is listed $m_j$ times.
It then follows directly from (\ref{eq:trace01}) that
\beq\label{eq:trace1}
Tr \left( \frac{ P_N (\theta) ( S_h (F, \theta) -
 \tilde{E}_N )^m P_N (\theta) }{ (h^2 \epsilon (h))^m } \right)
= \sum_{i=1}^{d_N} \left( \frac{ \nu_{N,i}}{h^2 \epsilon (h) } \right)^m .
\eeq


\subsection{Step 2. Evaluation of the trace.}

The second step in the proof of Theorem \ref{thm:trace1} is to estimate the trace on the left side
of (\ref{eq:trace1}). This requires that we replace $P_N(\theta)$ by $\Pi_N^0(\theta)$ and that
we control the perturbation $W_h(F, \theta)$ defined in (\ref{eq:pert1}).

In order to replace the projector $P_N(\theta)$ by $\Pi_N^0(\theta)$,
we need some results relating the projector $\Pi_N^0 (\theta)$ to $P_N (\theta)$
as $N \rightarrow \infty$, which means that $h = 1/N \rightarrow 0$.

\begin{lemma}\label{lemma:proj1}
For fixed $\theta \in \R$ with $0 < \theta < \pi / 3$,
we have
\beq\label{eq:proj-close1}
\| \Pi_N^0 ( \theta) - P_N (\theta) \| = \mathcal{O}( N^{6- K- \delta}),
\eeq
for $K \geq 3$ and $\delta > 0$.
Consequently, if $P^\perp \equiv 1 - P$, we have
\begin{enumerate}
\item $ \| (\Pi_N^0)^\perp (\theta) P_N(\theta) \| = \mathcal{O}( N^{6- K - \delta})$,
\item  $ \| P_N^\perp (\theta) \Pi_N^0 ( \theta) \| = \mathcal{O}( N^{6- K - \delta})$.
\end{enumerate}
\end{lemma}

\begin{proof}
\noindent
1. We consider the contour $\gamma_N$, a circle of radius
$1 / (8 N^3) > 0$ about the eigenvalue $\tilde{E}_N$. We write the
projectors as contour integrals
\beq\label{eq:contour1}
\Pi_N^0 (\theta) - P_N (\theta) = \frac{1}{2 \pi i} \int_{\gamma_N} ~dz [ (z
- H_V(\theta) )^{-1} - (z- S_h(F, \theta))^{-1} ] .
\eeq
Recall that $S_h(F, \theta) = H_V(\theta) + W_h(F, \theta)$,
and that $H_0(h,F , \theta) = -(1/2) e^{-2 i \theta} \Delta
+ W_h(F, \theta)$. We define a kernel (as in section \ref{sec:resonance1})
\beq\label{eq:kernel-res1}
K_h(F, \theta, z) = V(\theta) ( H_0(h, F, \theta) - z)^{-1}.
\eeq
From the second resolvent identity for $H_0 (h,F,\theta)$ and
$S_h(F, \theta)$,
we obtain
\beq\label{eq:resolv1}
(z- S_h(F, \theta))^{-1} = (z - H_0(h, F, \theta))^{-1} ( 1+ K_h (F, \theta, z))^{-1},
\eeq
whenever the inverse on the right exists.
Substituting this back into the second resolvent identity for $H_0 (h,F,\theta)$ and
$S_h(F, \theta)$, we obtain
\beq\label{eq:resolv2}
\frac{1}{z- S_h(F, \theta)} - \frac{1}{z - H_0(h,F,\theta)} =
 -\frac{1}{z - H_0(h,F,\theta)}\frac{1}{ 1+ K_h (F, \theta, z)} K_h (F, \theta, z).
\eeq
We note a similar identity for $h=0$ comparing $H_V(\theta) = H_0(\theta) + V(\theta)$
with $H_0(\theta) =-(1/2) e^{-2 i \theta} \Delta$. We let $K_0(\theta, z) \equiv
V(\theta)(H_0(\theta) - z)^{-1}$ in analogy with (\ref{eq:kernel-res1}) for $F=0$.
\beq\label{eq:resolv3}
\frac{1}{z- H_V (\theta)} - \frac{1}{z - H_0(\theta)} =
 -\frac{1}{z - H_0(\theta)} \frac{1}{ 1+ K_0 (\theta, z)} K_0 (\theta,z).
\eeq


\noindent
2. We subtract (\ref{eq:resolv2}) from (\ref{eq:resolv3}) and substitute the
difference into the integral in (\ref{eq:contour1}). Since both $(z-H_0(h,F,\theta))^{-1}$
and $(z-H_0(\theta))^{-1}$ are analytic on and inside
$\gamma_N$, their contribution to the contour integral vanishes.
Consequently, the difference of the projections in (\ref{eq:contour1})
is equal to the contour integral
\bea\label{eq:contour2}
\frac{1}{2 \pi i} \int_{\gamma_N} ~dz \left[ \frac{1}{z- H_0(h,F,\theta) }K_h(F, \theta,z)
\frac{1}{1 + K_h(F, \theta,z)}
 \right. \nonumber \\
 \left. - \frac{1}{z - H_0(\theta)} K_0 (\theta,z) \frac{1}{ 1+ K_0 (\theta,z)} \right] .
\eea
From Proposition \ref{prop:stark2}, part (2), we have that $K_h(F,\theta,z)$ and $K_0(\theta,z)$
are compact and that $K_h (F, \theta,z) \rightarrow K_0 (\theta,z)$ in norm.
We also use the fact that $(z-H_0(h,F,\theta))^{-1} \rightarrow (z - H_0(\theta))^{-1}$ strongly
as $h \rightarrow 0$.
We rewrite the integrand in (\ref{eq:contour2})
as a sum of three terms:
\bea 
I &\equiv & \frac{1}{z- H_0(h,F,\theta) }[ K_h(F, \theta,z) - K_0(\theta,z) ] \frac{1}{1 + K_h(F, \theta,z)}
  \label{eq:resterm1} \\
II &\equiv & \left[ \frac{1}{z- H_0(h,F,\theta) } - \frac{1}{z - H_0(\theta) } \right] K_0 (\theta,z)
\frac{1}{1 + K_0(\theta,z)} \label{eq:resterm2} \\
III &\equiv & \frac{1}{z- H_0(h,F,\theta) }K_0(\theta, z) \left[ \frac{1}{1 + K_h(F, \theta,z)} -
\frac{1}{1 + K_0 ( \theta, z)} \right]. \label{eq:resterm3}
\eea

\noindent
3. We need estimates on the operators $K_0 (\theta, z)$ and $K_h (F , \theta,z)$ and their resolvents at $-1$ that
appear in \eqref{eq:resterm1}, \eqref{eq:resterm2}, and \eqref{eq:resterm3}.
We begin with the estimates for $K_0(\theta, z)$. From part (3) of Lemma \ref{lemma:resolventest1}, we have
\beq\label{eq:k0-1}
\| K_0 (\theta, z) \| = \mathcal{O}(N).
\eeq
From the definition of $K_0 (\theta, z)$, we easily find that
\beq\label{eq:K0-2}
(1 + K_0 (\theta, z))^{-1} = 1 - V(\theta) ( H_0 (\theta) + V(\theta) - z)^{-1} .
\eeq
Recall that the Coulomb potential $V$ is relatively $H_0$-bounded with relative bound less than one. So there
exist constants $0 < a < 1$ and $b > 0$, so that for all $u \in H^2 (\R^3)$, we have
\beq\label{eq:relVbd1}
\| V u \| \leq a \| H_0 u \| + b \| u \| .
\eeq
Scaling by $e^{i \theta}$, with $\theta \in \R$, we find
\beq\label{eq:relVbd10}
\| V(\theta) u \| \leq a \| H_0 (\theta) u \| + b \| u \| .
\eeq
Replacing $u$ by $R_V(\theta) w = (H_0 (\theta) + V(\theta) - z)^{-1}w$, for $z \in \gamma_N$
and $w \in L^2 (\R^3)$, we obtain
\beq\label{eq:relVbd2}
\| V(\theta) R_V(\theta) w \| \leq a \| w\| + ( a |z| + b) \| R_V(\theta) w \| + a \| V(\theta) R_V(\theta) w \| .
\eeq
Since $0 < a < 1$ and $|z| = \mathcal{O}(N^{-2})$, we obtain
\beq\label{eq:relVbd3}
\| V(\theta) R_V(\theta) \| \leq C_1 + C_2 \| R_V(\theta)  \| .
\eeq
It follows from \eqref{eq:K0-2}, \eqref{eq:relVbd3}, and the fact that $\| R_V(\theta) \| = \mathcal{O}(N^3)$ that
\beq\label{eq:K0-3}
\| (1 + K_0 (\theta, z))^{-1} \| = \mathcal{O}(N^3).
\eeq

\noindent
4. The second estimate we need concerns $1 + K_h (F, \theta,z)$. We write
\bea\label{eq:KF-1}
1 + K_h (F, \theta, z) &=& [ 1 + K_0 ( \theta, z) ] \nonumber \\
 && \times [  1 +  ( 1 + K_0  ( \theta, z) )^{-1} ( K_h (F, \theta,z) - K_0 (\theta, z) )] ,
 \eea
 from which it follows that
 \beq\label{eq:KF-2}
( 1 + K_h (F, \theta,z))^{-1} = [ 1 + M_h (F, \theta,z) ]^{-1} [ 1 + K_0 ( \theta, z) ]^{-1},
\eeq
where
\beq\label{eq:M-defn1}
M_h(F , \theta, z) = ( 1 + K_0 ( \theta, z) )^{-1} ( K_h (F, \theta, z) - K_0 (\theta, z)).
\eeq
It follows from part (2) of Proposition \ref{prop:stark2} and \eqref{eq:K0-3} that
\beq\label{eq:M-bound1}
\| M_h (F, \theta,z) \| = \mathcal{O}(N^{3-K-\delta}),
\eeq
so if $K \geq 3$, this term is less than $1/2$ for all $N$ large. It follows from
\eqref{eq:KF-2} that
\beq\label{eq:KF-3}
\| ( 1 + K_h (F, \theta,z))^{-1}  \| = \mathcal{O}(N^3).
\eeq

\noindent 5. We estimate each term I, I, and III in
\eqref{eq:resterm1}, \eqref{eq:resterm2}, and \eqref{eq:resterm3},
uniformly in $z \in \gamma_N$, using the estimate of Proposition \ref{prop:stark2}.
For the first term I in (\ref{eq:resterm1}), we use part (1) of Lemma \ref{lemma:resolventest1},
part (2) of Proposition \ref{prop:stark2}, and \eqref{eq:KF-3}, to obtain
for $K \geq 3$:
\beq\label{eq:estI-1}
\| I \| = \mathcal{O}( N^2 \cdot N^{-K - \delta} \cdot N^3) = \mathcal{O}(N^{5-K-\delta}).
\eeq
As for II in (\ref{eq:resterm2}), we use the quadratic estimate in \cite[Proposition II.4]{herbst1}
(presented in appendix \ref{appendix:QE})
in order to prove the bound
\beq\label{eq:QE1}
\| x_1 ( H_0(h,F, \theta) - z)^{-1} \| = \mathcal{O}(1), ~~ z \in \gamma_N.
\eeq
Recalling that $H_0 (h, F, \theta) - H_0 (\theta) = W_h (F, \theta) = h^4 \epsilon (h ) e^{i \theta} F x_1$,
we find that
\beq\label{eq:estII-1}
\| II \| = \mathcal{O}(  N^{-4 -K - \delta}  \cdot N^2 \cdot N \cdot N^3) = \mathcal{O}(N^{2-K-\delta}).
\eeq
Recalling that $K \geq 3$, we see
that the term II vanishes uniformly on $\gamma_N$ as $N \rightarrow \infty$.

Finally, for the last term III in (\ref{eq:resterm3}),
part (1) of Lemma \ref{lemma:resolventest1}, estimate \eqref{eq:k0-1}, together with
\eqref{eq:KF-3}, \eqref{eq:K0-3} and part (2) of Proposition \ref{prop:stark2}, yield
\bea\label{eq:estIII-1}
\| III \| &= & \mathcal{O}( N^2 \cdot N \cdot N^3 \cdot N^{-K - \delta} \cdot N^3) \nonumber \\
 &=& \mathcal{O}(N^{9-K-\delta}).
\eea


\noindent
6. The difference of the projectors on the left side of \eqref{eq:proj-close1}
may be estimated from \eqref{eq:contour2} and the above estimates, recalling that $|\gamma_N| = 2 \pi (1/ (8N^3))$:
\beq\label{eq:proj-close2}
\| \Pi_N^0 ( \theta) - P_N (\theta) \| \leq | \gamma_N | ( \| I \| + \| II \| + \| III \| )
  =  \mathcal{O}(N^{6-K-\delta}), ~~ K \geq 3.
\eeq
So for $K \geq 6$, we obtain the vanishing of the difference as $N \rightarrow \infty$.
Parts (1) and (2) of the lemma follow from (\ref{eq:proj-close1}) simply by writing
\beq\label{eq:perp-p1}
(\Pi_N^0 (\theta))^\perp P_N(\theta) =  (\Pi_N^0 (\theta))^\perp ( P_N (\theta) - \Pi_N^0 (\theta) ) ,
\eeq
and similarly for part (2).
\end{proof}

\subsection{Step 3. Estimates on dilated coherent states.}\label{subsec:dlated-cohst1}

In order to control the perturbation $W_h (F, \theta)$, we need estimates on the
following operators: $\Pi_N^0 (\theta) W_h(F,\theta) \Pi_N^0 (\theta)$,
$\Pi_N^0 (\theta) W_h(F,\theta)$, and the operators $P_N(\theta) (S_h(F, \theta) - \tilde{E}_N) P_N(\theta)$
and $P_N(\theta) (S_h(F, \theta) - \tilde{E}_N) \Pi_N^0 (\theta)$. Estimates on the first two
operators are given in Lemma \ref{lemma:moments1}, and on the second two operators
are given in Lemma \ref{lemma:pretrace1}.

Heuristically, we are able to control the first two operators due to the fact that the matrix elements of moments of the position
operator $\|x\|$ in the eigenstates $\psi_N(\theta)$ of $H_V(\theta)$ satisfy
\beq\label{eq:matrixele0}
\langle \psi_N (\theta) , \|x\|^m \psi_N (\theta) \rangle \sim N^{2m} .
\eeq
This decay is due to the fact that the eigenstate is well localized about the Bohr radius and the Bohr
radius scales like $N^2$.
A proof of this localization property of the eigenfunctions is given in the Appendix 2
in section \ref{sec:app2-ef}.
This localization, however, is too weak to control the operator norm of the operators $\Pi_N^0 (\theta) W_h(F,\theta) \Pi_N^0 (\theta)$ and
$\Pi_N^0 (\theta) W_h(F,\theta)$. Using \eqref{eq:matrixele0}, we easily arrive at estimates of the type
\beq\label{eq:efest1}
\| \Pi_N^0 (\theta) W_h(F, \theta) \Pi_N^0 (\theta) \| = \mathcal{O}( N^2 \epsilon (h) h^2 ),
\eeq
and the $N^2$ growth will not allow control of the trace in (\ref{eq:trace2}) since this is divided by
$h^2 \epsilon (h)$. Instead of using \eqref{eq:matrixele0},
we use coherent states $\Psi_{\alpha, N}$ that
form an overdetermined basis set of functions for the eigenspace of the hydrogen atom Hamiltonian
corresponding to the eigenvalue $\tilde{E}_N$. These coherent states were described in detail
in the papers of Bander and Itzykson \cite{BI:1966} and more recently in \cite{thomas-villegas1,uribe-villegas1}.
We recall the main points that we need here in Appendix 1, section \ref{sec:app1-coherentstates1}.
Recall that the dimension of the range of the projector $\Pi_N^0 (\theta)$ is $N^2$ and that $\epsilon (h)
= h^{K + \delta}$, for $K \geq  6$.

We use the following notation for operators that occur as remainder terms
and that have bounds depending on $h$ but uniform
with respect to any other parameters.
Let $K ( g(h)) $, for a function $g(h)$,
denote a bounded operator
with
\beq\label{eq:op-bound1}
\| K ( g(h)) \| = \mathcal{O}(g(h)).
\eeq
an example of a function $g(h)$ is $(\epsilon (h) h^2)^m$.
We will also write $K( \mathcal{O}(h^{-\ell}))$ to mean a bounded linear operator
with $\|K( \mathcal{O}(h^{-\ell})) \| = \mathcal{O}(h^{-\ell})$.
The actual form of $K$ is unimportant and may vary from line
to line but a bound of the type \eqref{eq:op-bound1} or of the type $\mathcal{O}(h^{-\ell})$ will always hold.

\begin{lemma}\label{lemma:moments1}
There exists a constant $r_0 > 1$, independent of $\theta$ with $| \theta| < \pi / 4$, so that for any
$n \in \N$, we have
\bea\label{eq:cost-proj01}
\Pi_N^0 (\theta) D_{N^2}^{-1} \|x\|^n D_{N^2} \Pi_N^0 (\theta) & =& \Pi_N^0 (\theta) D_{N^2}^{-1} \|x\|^n
\chi_{\|x\| \leq r_0} D_{N^2} \Pi_N^0 (\theta) \nonumber \\
  & & + \Pi_N^0 K( \mathcal{O}(N^{-\infty})),
\eea
where $\chi_{\|x\| \leq r_0}$ is the characteristic function on the set $\{ x \in \R^3 ~|~ \| x \| \leq r_0 \}$.
As a consequence, we have the following estimates on the perturbation restricted to the
eigenspace of $H_V(\theta)$:
\begin{enumerate}
\item $ \| \Pi_N^0 (\theta) W_h(F, \theta) \Pi_N^0 (\theta) \| = \mathcal{O}(N^{-K-2-\delta})
= \mathcal{O}(\epsilon (h) h^2 )$ ,
\item  $ \| \Pi_N^0 (\theta) W_h(F, \theta) \| = \mathcal{O}(N^{-K-2-\delta})
= \mathcal{O}( \epsilon (h) h^2 )$.
\end{enumerate}
\end{lemma}

\begin{proof}
As mentioned above, the key to controlling the perturbation is the strong localization
property of the coherent states. Coherent states for the hydrogen atom are reviewed in section \ref{sec:app1-coherentstates1}.
We prove below that the dilated coherent states $\Psi_{\alpha, N} (e^{i \theta} x)$
are $L^2$-valued analytic functions of $\theta$ provided $| \Re \theta| < \pi / 4$
and provide uniform bounds.
Since, as above, $\Im \theta$ pays no role in the calculations, we set $\Im \theta = 0$.

\noindent
1. We first prove a decay estimate for the dilated coherent states that is the
analog of \cite[Lemma 4.1]{thomas-villegas1}.
We prove that there exists a constant $r_0 > 1$, independent of $\alpha \in \mathcal{A}$, $N \in \N$,
and $| \theta | < \pi / 4$, so that for all $n,s \in \N$,
\beq\label{eq:cost-limit1}
\lim_{N \rightarrow \infty} N^s \int_{\| x \| > r_0 } ~\|x  \|^n | D_{N^2} \Psi_{\alpha, N} (e^{i \theta} x) |^2 ~d^3 x = 0.
\eeq
This estimate implies that
\beq\label{eq:cost-limit11}
\left| \int_{\R^3}  \Psi_{\beta, N}^* (e^{i \theta} x) D_{N^2}^{-1} \| x \|^n \chi_{\|x\| >r_0} D_{N^2}
\Psi_{\alpha, N} (e^{i \theta} x) ~d^3 x \right| = \mathcal{O}(N^{- \infty} ).
\eeq
This estimate is uniform with respect to $\alpha, \beta \in \mathcal{A}$.
We turn to the proof of \eqref{eq:cost-limit1}.
For any $b \in \Sp^2$, we will prove
\beq\label{eq:cost-est1}
\lim_{N \rightarrow \infty} N^s \int_{| x \cdot b| \geq r_0 / 2} ~ |x \cdot b|^n
| D_{N^2} \Psi_{\alpha, N} (e^{i \theta} x)|^2 ~d^3 x = 0.
\eeq
By making appropriate choices of $b$ we recover \eqref{eq:cost-limit1}
by a finite sum.
Following the proof of \cite[Lemma 4.1]{thomas-villegas1},
we first prove that for any $q \geq 0$, there exist constants $c_2 , c_3 > 0$,
independent of $\theta$, $\alpha$ and $N$,
so that
\beq\label{eq:cost-est1a}
\left( \int_{\R^3}  |x \cdot b|^{2q}
| D_{N^2} \Psi_{\alpha, N} (e^{i \theta} x)|^2 ~d^3 x \right)^{1/2 }   \leq
\frac{c_3 q! N^{1/2} e^{c_2 N }}{ (N/ 240)^q} ,
\eeq
for any $b \in \Sp^2$. We repeat the argument from \cite[Lemma 4.1]{thomas-villegas1} showing how
the estimate \eqref{eq:cost-est1a} implies \eqref{eq:cost-est1}.
Given $n$ and $s$ from \eqref{eq:cost-est1}, we take $m \in \N$ so that $N-1 < m+n <N$ and set $q = (n+m) / 2$.
Since $|x \cdot b| \geq r_0 / 2$,
we use Chebyshev's inequality and Stirling's formula for $q!$
to estimate \eqref{eq:cost-est1} from \eqref{eq:cost-est1a},
\bea\label{eq:cost-est1b}
\lefteqn{N^s \int_{| x \cdot b| \geq r_0 / 2} ~ |x \cdot b|^n
| D_{N^2} \Psi_{\alpha, N} (e^{i \theta} x)|^2 ~d^3x } \nonumber \\
 &\leq& \frac{N^s}{ (r_0 /2)^m} \int_{\R^3 } ~ |x \cdot b|^{m+n}
| D_{N^2} \Psi_{\alpha, N} (e^{i \theta} x)|^2 ~d^3x  \nonumber \\
&\leq& \frac{N^s}{ (r_0 /2)^m} \left[ c_3 q! N^{1/2} (N/ 240)^{-q} e^{N c_2 } \right]^2 \nonumber \\
 & \leq & \frac{c_5 N^{s + 3}}{ (r_0 / 2)^{N - 1 - n}} e^{ N  c_4 } ,
\eea
where $c_4 \equiv 2  c_2 + \log 120 -1 > 0$.
We choose $r_0 > 2$ so that $r_0 / 2 > e^{2 c_4}$ so that the
right side of \eqref{eq:cost-est1b} is bounded by
\beq\label{eq:cost-est1c} c_6 N^{s + 3} e^{- Nc_4 } , \eeq and
this vanishes as $N \rightarrow \infty$. This proves
\eqref{eq:cost-est1}.

\noindent
2. To prove (\ref{eq:cost-est1a}), we change to the momentum variable (see \eqref{eq:cohst-p1})
so that
\bea\label{eq:cost-est2}
\lefteqn{ \| (x \cdot b)^q D_{N^2} \Psi_{\alpha, N} (e^{i \theta} \cdot) \|^2 } \nonumber \\
& = & \int_{\R^3} \left| \left(
\frac{1}{N} b \cdot \nabla_p \right)^q \left( \frac{2}{ e^{-2 i \theta} \|p \|^2 +1 } \right)^2 a(N-1) (\alpha \cdot \omega ( e^{-i \theta} p))^{N-1}
\right|^2 ~d^3 p . \nonumber \\
 & &
\eea
We next use the fact that $b \cdot \nabla_p$ generates translations in $p$ so that for $z \in \C$,
\bea\label{eq:cost-est3}
\lefteqn{ e^{ (z/N) b \cdot \nabla_p} \left( \frac{2}{ e^{-2 i \theta} \|p \|^2 +1 } \right)^2 a(N-1) (\alpha \cdot \omega ( e^{-i \theta} p))^{N-1}
 } \nonumber \\
 &=& \left( \frac{2}{ e^{-2 i \theta} \left(p + (z/N) b \right)^2 +1 } \right)^2 a(N-1) (\alpha \cdot \omega ( e^{-i \theta}( p + (z/N)b ))^{N-1}
\eea
We need some estimates. First, in order to guarantee that the
function in \eqref{eq:cost-est3} remains in $L^2 (\R^3)$, we observe
that if $|z| / N < 1/120$ and $|\theta| < \pi / 4$, there are finite constants
$0 < C_1, C_2$ so that
\beq\label{eq:cost-est5a}
\left| \frac{2 e^{2 i \theta}}{  (p + (z/N) b)^2   + e^{2 i \theta}} \right|^2 \leq
\left\{ \begin{array}{cc}
 \left( \frac{C_1}{\|p \|^2 + 1 } \right)^2 & \| p \| > 2 \\
  C_2 & \| p \| \leq 2 .
\end{array}
\right.
\eeq
This estimate is proved by estimating the absolute values of the real and imaginary parts of
$( p + (z/N) b )^2 + e^{2 i \theta}$ from below. So provided all other factors are
uniformly bounded in $p$, the function in \eqref{eq:cost-est3} is
square integrable. Next, we prove the uniform bounds on the other
factors for $|z| / N < 1/2$ and $|\theta| < \pi / 3$. These conditions are less restrictive than needed for \eqref{eq:cost-est5a}.
 We note that
\beq\label{eq:cost-est4}
 | \alpha \cdot \omega ( e^{-i \theta}p) |
\leq  \sqrt{10} ,
\eeq
for $|\theta| < \pi / 3$ since $| \alpha| =
\sqrt{2}$.
In order to estimate $\omega$, we expand about $p$ and write
\beq\label{eq:cost-est5} \omega ( e^{-i \theta}(p + (z/N) b)) =
\omega ( e^{-i \theta}p) + \nabla_p \omega ( e^{-i
\theta}\tilde{p}) \cdot (z/N) e^{-i \theta} b, \eeq for some
$\tilde{p}$. It is easy to check that for $|z|/ N < 1/2$, the
gradient term satisfies \beq\label{eq:cost-est6} |\nabla_p \alpha
\cdot \omega ( e^{-i \theta}(p+ (z/N)b )) | \leq c_1, \eeq so that
\beq\label{eq:cost-est7} | \alpha \cdot \omega ( e^{-i \theta}(p +
(z/N) b)) | \leq  \sqrt{10}( 1 + c_2 |z| / N) . \eeq Consequently,
for any $N$ and $z \in \C$ so that $|z| / N < 1/2$, we have
\beq\label{eq:cost-est8} | \alpha \cdot \omega ( e^{-i \theta}(p +
(z/N) b)) |^{N-1} \leq  c_0^N e^{c_2 |z|} , \eeq for absolute
constants $c_0, c_1 > 0$. Combining these, we obtain
\beq\label{eq:cost-est9} \left\| e^{ (z/N) b \cdot \nabla_p}
\left( \frac{2}{ e^{-2 i \theta} \|p \|^2 +1 } \right)^2 a(N-1)
(\alpha \cdot \omega ( e^{-i \theta} p))^{N-1} \right\| \leq c_3
N^{1/2} c_0^N e^{c_2 | z|} , \eeq since $a(N-1) \sim \sqrt{N-1}$.

\noindent
3. We now use Cauchy's theorem, with estimate (\ref{eq:cost-est9}), in order to estimate (\ref{eq:cost-est2}), by integrating
over a path in the $z$-plane of radius $N / 240 < N / 120$ about the origin so that estimate \eqref{eq:cost-est5a} is valid.
This gives
\bea\label{eq:cost-est10}
\lefteqn{ \| (x \cdot b)^q D_{N^2} \Psi_{\alpha, N} (e^{i \theta} \cdot ) \| } \nonumber \\
 & \leq & \frac{q!}{2 \pi} \int_{|z| = N/240} \frac{|dz|}{|z|^{q+1}}
\left\| e^{ (z/N) b \cdot \nabla_p} \left( \frac{2}{ e^{-2 i \theta} \|p \|^2 +1 } \right)^2 a(N-1) (\alpha \cdot \omega ( e^{-i \theta} p))^{N-1}
\right\|  \nonumber \\
 & \leq &  \frac{c_3 q! N^{1/2} e^{c_2 N }}{ (N/ 240)^q} ,
\eea
where the finite constant $c_2 >0$  is a function of $c_0$ and $c_1$.
This establishes \eqref{eq:cost-est1a}.

\noindent
4. It follows from estimate (\ref{eq:cost-limit1}) that for $\alpha, \beta \in \mathcal{A}$, we have
\beq\label{eq:cost-me1}
\langle \Psi_{\alpha, N}, D_{N^2}^{-1} \|x\|^n \chi_{\|x\| > r_0} D_{N^2} \Psi_{\beta, N} \rangle = \mathcal{O}(N^{-\infty}),
\eeq
where the error is uniform over $\mathcal{A} \times \mathcal{A}$.
Of importance for us is that this estimate \eqref{eq:cost-me1} implies that the moments of the position operator in coherent states
satisfy
\bea\label{eq:cost-me2}
\lefteqn{ \langle \Psi_{\alpha, N}( \cdot; \theta), D_{N^2}^{-1} \|x\|^n  D_{N^2} \Psi_{\beta, N}(\cdot; \theta) \rangle } \nonumber \\
& = &
\langle \Psi_{\alpha, N} (\cdot; \theta), D_{N^2}^{-1} \|x\|^n \chi_{\|x\| \leq r_0} D_{N^2} \Psi_{\beta, N} (\cdot; \theta) \rangle
+  \mathcal{O}(N^{-\infty}).
\eea
It now follows from (\ref{eq:cost-me2}) and the representation \eqref{eq:proj-coherentstates1} of the projector, suitably dilated, that
we have the operator estimate
\beq\label{eq:cost-proj1}
\Pi_N^0 (\theta) D_{N^2}^{-1} \|x\|^n D_{N^2} \Pi_N^0 (\theta)  =
 \Pi_N^0 (\theta) D_{N^2}^{-1} \|x\|^n
\chi_{\|x\| \leq r_0} D_{N^2} \Pi_N^0 (\theta) + \Pi_N^0 ( \theta ) R_N \Pi_N^0 (\theta),
\eeq
%
%
where the remainder $R_N$ is given by
\beq\label{eq:cost-remainder1}
R_N \equiv \int_\mathcal{A} \int_\mathcal{A}
\langle \Psi_{\alpha, N}( \cdot; \theta), D_{N^2}^{-1} \|x\|^n \chi_{ \| x\| > r_0}
D_{N^2} \Psi_{\beta, N}(\cdot; \theta) \rangle ~P_{\alpha, \beta}~d \mu (\alpha) d \mu (\beta) ,
\eeq
where $P_{\alpha, \beta}$ is the dyadic operator
\beq\label{eq:dyadic1}
P_{\alpha, \beta} \equiv ~ | \Psi_{\alpha, N} (\theta) \rangle
\langle \Psi_{\beta, N} (\theta) |  .
\eeq
To estimate $\| R_N \|$, we use estimate \eqref{eq:cost-me1} and the
fact that the measure $\mu$ on $\mathcal{A}$ is a probability measure.
We obtain
\beq\label{eq:cost-remainder2}
\| R_N \| \leq C \sup_{\alpha, \beta \in \mathcal{A}} ( \| \Psi_{\alpha, N}(\theta) \| ~  \|
\Psi_{\beta, N}(\theta) \| ) e^{- N c_4 } ,
\eeq
where $c_4 > \log 240$ as in \eqref{eq:cost-est1c}.
The $L^2$-norms of the dilated coherent states can be estimates using \eqref{eq:cost-est5a} and
\eqref{eq:cost-est4}. They satisfy the bound
\beq\label{eq:cost-dilate-bd1}
\| \Psi_{\alpha, N}(\theta) \| \leq C N^2 e^{N (\log 10 ) / 2} , ~~| \Im \theta| < \pi / 2.
\eeq
Consequently, it follows from \eqref{eq:cost-remainder2} and \eqref{eq:cost-dilate-bd1}
that $\| R_N \| \leq C^{- N c_5}$, for some $c_5 > 0$.
Equation \eqref{eq:cost-proj01} then follows from \eqref{eq:cost-proj1}, \eqref{eq:cost-remainder2}, and the fact
that the measure $\mu$ on $\mathcal{A}$ is a probability measure

\noindent
5. We can now prove the lemma. Recall from (\ref{eq:pert1}) that
\beq\label{eq:pert11}
W_h (F, \theta) = h^4 \epsilon (h) e^\theta F x_1 = h^2 \epsilon (h) D_{N^2}^{-1} (Fx_1) D_{N^2},
\eeq
for $N = 1/h$. For part (1), we have
\beq\label{eq:moment2}
\| \Pi_N^0(\theta) W_h(F,\theta) \Pi_N^0 (\theta) \| \leq c_0 (h^2 \epsilon (h))F r_0 + \mathcal{O}( N^{-\infty}),
\eeq
for a constant $r_0> 1$.
For part (2), we use $\| ( \Pi_N^0(\theta) W_h(F,\theta) )^* ( W_h(F,\theta) \Pi_N^0 (\theta) ) \|
= \| W_h(F,\theta) \Pi_N^0 (\theta) \|^2 =
 \| \Pi_N^0(\theta) W_h(F,\theta) \|^2$, so that estimate (\ref{eq:cost-proj1}) with $n=2$ provides the estimate.
\end{proof}


\subsection{Step 4. Reduction of the perturbation.}\label{subsec:perturbation1}

We now turn to controlling the operators $P_N(\theta) (S_h(F, \theta) - \tilde{E}_N) P_N(\theta)$
and $P_N(\theta) (S_h(F, \theta) - \tilde{E}_N) \Pi_N^0 (\theta)$.
Using Lemmas \ref{lemma:proj1} and \ref{lemma:moments1}, we can prove the analog of
\cite[Lemma 5]{uribe-villegas1}.

\begin{lemma}\label{lemma:pretrace1}
For any positive integer $m$, we have
\beq\label{eq:offdiag1}
\Pi_N^0 (\theta) \left( \frac{S_h(F, \theta) - \tilde{E}_N }{ \epsilon (h) h^2} \right)^m P_N (\theta)
 = ( \Pi_N^0 (\theta) \tilde{W}_h (F, \theta) \Pi_N^0 (\theta) )^m P_N(\theta) + \Pi_N^0 (\theta) R_{m,N},
 \eeq
where $\| R_{m,N} \| = \mathcal{O}(N^{- \beta}) $, for some $\beta > 0$,
independent of $m$, and $\tilde{W}_h (F, \theta)
= e^{i\theta} h^2 F x_1$, with $\theta \in \R$ and $0 < |\theta| < \pi / 4$.
\end{lemma}

\begin{proof}
\noindent
1. To simplify notation, we suppress the $\theta$ in the notation.
We begin with a simple identity:
\beq\label{eq:reduction-id1}
\Pi_N^0 ( S_h - \tilde{E}_N ) P_N  =  \Pi_N^0
( S_h - \tilde{E}_N ) \Pi_N^0 P_N + \Pi_N^0 ( S_h - \tilde{E}_N ) (\Pi_N^0)^\perp P_N .
\eeq
We need an identity that follows from analyticity.
Since $\tilde{E}_N$ remains an eigenvalue of the
dilated hydrogen atom Hamiltonian, $H_V(\theta)$, we have
\beq\label{eq:analytic1}
\Pi_N^0 (\theta) ( H_V(\theta) - \tilde{E}_N ) \Pi_N^0 (\theta) = 0.
\eeq
Using this (\ref{eq:analytic1}), and the fact that $\Pi_N^0 H_V (\Pi_N^0)^\perp = 0$, we have
\bea\label{eq:reduction-id2}
\Pi_N^0 ( S_h - \tilde{E}_N ) \Pi_N^0  & = & \Pi_N^0 W_h \Pi_N^0  \\
\Pi_N^0 ( S_h - \tilde{E}_N ) (\Pi_N^0)^\perp & = & \Pi_N^0 W_h (\Pi_N^0)^\perp .
\eea
Substituting these into (\ref{eq:reduction-id1}), we obtain
\beq\label{eq:reduction-id3}
\Pi_N^0 ( S_h - \tilde{E}_N ) P_N  =  \Pi_N^0 W_h \Pi_N^0 P_N + \Pi_N^0 W_h (\Pi_N^0)^\perp P_N .
\eeq

\noindent
2. To estimate the second term on the right in (\ref{eq:reduction-id3}),
we use part (1) of Lemma \ref{lemma:proj1}
and part (2) of Lemma \ref{lemma:moments1}:
\beq\label{eq:reduction-id4}
\| \Pi_N^0 W_h (\Pi_N^0)^\perp P_N \| \leq \| \Pi_N^0 W_h \| ~\| (\Pi_N^0)^\perp P_N \| =
\mathcal{O}( N^{4-2K - \delta}). 
\eeq
We take the $m^{th}$ power of (\ref{eq:reduction-id3}) and, because of (\ref{eq:reduction-id4}), we have
\beq\label{eq:reduction-id5}
( \Pi_N^0 ( S_h - \tilde{E}_N ) P_N )^m  =  ( \Pi_N W_h \Pi_N^0 P_N )^m + \Pi_N^0 \tilde{R}_{m,N} ,
\eeq
where the error term $\tilde{R}_{m,N}$ has the form
\beq\label{eq:error-1}
\tilde{R}_{m,N} = \sum_{\ell=1}^m \left( \begin{array}{c}
 m \\ \ell \end{array} \right) \| \Pi_N^0 W_h \Pi_N^0 P_N \|^{m-\ell}
\| \Pi_N^0 W_h (\Pi_N^0)^\perp P_N \|^\ell,
\eeq
From Lemma \ref{lemma:proj1} part (1) and Lemma \ref{lemma:moments1},
we obtain
\beq\label{eq:error-2}
\| \Pi_N^0 W_h \Pi_N^0 P_N \|^{m-\ell} \| \Pi_N^0 W_h (\Pi_N^0)^\perp P_N \|^\ell \leq C(\theta)
(\epsilon(h) h^2)^m \mathcal{O}( N^{-\delta} ).
\eeq
Consequently, $\tilde{R}_{m,N}$ satisfies the estimate
\beq\label{eq:reduction-id6}
\| \tilde{R}_{m,N} \| \leq
(\epsilon (h) h^2)^m \mathcal{O}(N^{-\delta}).
\eeq

\noindent
3. We next prove that
\beq\label{eq:reduction-id7}
( \Pi_N^0 W_h \Pi_N^0 P_N )^m   = ( \Pi_N^0 W_h \Pi_N^0)^m  P_N + \Pi_N^0 K ((\epsilon (h) h^2)^m \mathcal{O}(N^{-\delta})) ,
\eeq
for all $m \in \N$ by induction on $m$, where we use the notation
$K$ introduced before Lemma \ref{lemma:moments1}.
We proceed by induction. Equality \eqref{eq:reduction-id7} is trivially true for $m=1$. We assume it is true for $m-1$ and verify it for $m$. We write
\bea\label{eq:interact-est1}
( \Pi_N^0 W_h \Pi_N^0 P_N )^m &=& ( \Pi_N^0 W_h \Pi_N^0 P_N ) ( \Pi_N^0 W_h \Pi_N^0 P_N )^{m-1} \nonumber \\
 &=& (  \Pi_N^0 W_h \Pi_N^0 -  \Pi_N^0 W_h \Pi_N^0 P_N^\perp) [ ( \Pi_N^0 W_h \Pi_N^0 )^{m-1} P_N  \nonumber \\
 & & + \Pi_N^0 K( (h^2 \epsilon(h))^{m-1} \mathcal{O} (N^{-\delta}))]
\nonumber \\
 &=& ( \Pi_N^0 W_h \Pi_N^0  )^m P_N + ( \Pi_N^0 W_h \Pi_N^0 ) K( (h^2 \epsilon(h))^{m-1} \mathcal{O}( N^{-\delta}) ) \nonumber \\
 & & -  \Pi_N^0 W_h \Pi_N^0 P_N^\perp ( \Pi_N^0 W_h \Pi_N^0 )^{m-1} P_N \nonumber \\
 & & - \Pi_N^0 W_h \Pi_N P_N^\perp \Pi_N^0 K( (h^2 \epsilon(h) )^{m-1} \mathcal{O} (N^{-\delta})) .
\eea
Using the estimates in Lemma \ref{lemma:moments1} for $\Pi_N^0 W_h \Pi_N^0$, we establish (\ref{eq:reduction-id7}).

\noindent
4. We next prove a similar estimate
\beq\label{eq:reduction-id8}
( \Pi_N^0 (S_h - \tilde{E}_N ) P_N)^m = \Pi_N^0 (S_h -
\tilde{E}_N)^m P_N +  \Pi_N^0 (\epsilon (h) h^2)^m \mathcal{O}(N^{-\delta}).
\eeq
In order to estimate the resonance term $(S_h - \tilde{E}_N) P_N$, we write
\beq\label{eq:res-est11}
(S_h - \tilde{E}_N ) P_N = (S_h - \tilde{E}_N ) P_N (P_N - \Pi_N^0) + P_N W_h \Pi_N^0 .
\eeq
Noting that $\| \Pi_N^0 - P_N \| < 1$, we have from \eqref{eq:res-est11},
\beq\label{eq:res-est22}
P_N (S_h - \tilde{E}_N) P_N = P_N W_h \Pi_N^0 ( 1 + (\Pi_N^0 - P_N))^{-1} .
\eeq
From Lemma \ref{lemma:proj1}, we obtain the estimate
\beq\label{eq:res-est33}
\|   (S_h - \tilde{E}_N) P_N \| \leq c_0 h^2 \epsilon (h) .
\eeq
Given estimate \eqref{eq:res-est33}, we prove \eqref{eq:reduction-id8} by induction. Assuming \eqref{eq:reduction-id8} for $m-1$, we
write
\bea\label{eq:res-est44}
\lefteqn{ ( \Pi_N^0 ( S_h - \tilde{E}_N ) P_N)^m - \Pi_N^0 ( S_h - \tilde{E}_N )^m P_N } \nonumber \\
&=& \Pi_N^0 ( S_h - \tilde{E}_N ) \left[ P_N ( \Pi_N^0 ( S_h - \tilde{E}_N ) P_N)^{m-1} -
 \right. \nonumber \\
 & & \left. - ( S_h - \tilde{E}_N )^{m-1} P_N \right]
\eea
From \eqref{eq:res-est33}, we have the bound
\beq\label{eq:res-est55}
\| (S_h - \tilde{E}_N )^{m-1} P_N \| \leq \|  [(S_h - \tilde{E}_N ) P_N ]^{m-1} \| \leq ( c_0 h^2 \epsilon(h))^{m-1} .
\eeq
Consequently, the norm of the left side of \eqref{eq:res-est44} may be bounded above by
\beq\label{eq:res-est66}
\| \Pi_N^0 ( S_h - \tilde{E}_N ) P_N [ - P_N (\Pi_N^0)^\perp ( S_h - \tilde{E}_N )^{m-1} P_N +  P_N
\Pi_N^0 K( (h^2 \epsilon(h) )^{m-1} \mathcal{O}(N^{-\delta})) ] \|.
\eeq
The estimate \eqref{eq:reduction-id8} for $m$ now follows from this and \eqref{eq:res-est33} and \eqref{eq:res-est55}.
This completes the proof of Lemma \ref{lemma:pretrace1}
\end{proof}


\subsection{Completion of the proof of Theorem \ref{thm:trace1}.}\label{subsec:completion-th4}

In order to estimate the trace of $(S_h(F) - \tilde{E}_N)^m P_N$ on the left side in (\ref{eq:trace1}),
we write
\beq\label{eq:decomp1}
 (S_h(F) - \tilde{E}_N)^m P_N =  \Pi_N^0 (S_h(F) - \tilde{E}_N)^m P_N + (\Pi_N^0)^\perp (S_h(F) - \tilde{E}_N)^m P_N .
\eeq
Due to Lemma \ref{lemma:pretrace1},
we have
\bea\label{eq:decomp2}
(S_h(F) - \tilde{E}_N)^m P_N  - (\Pi_N^0 W_h(F) \Pi_N^0 )^m  &= &
 (\Pi_N^0)^\perp (S_h(F) - \tilde{E}_N)^m P_N \nonumber \\
 & & + (\epsilon(h) h^2 )^m \Pi_N^0 R_{m,N} \nonumber \\
 & &  - (\Pi_N^0 W_h(F) \Pi_N^0 )^m P_N^\perp  .
\eea
We estimate the trace norm of each term on the right in (\ref{eq:decomp2}).

For the first term, we use the fact that $\| P_N \|_1 = d_N$, part (1)
of Lemma \ref{lemma:proj1}, and estimates on resonances in order to estimate
$\| ( S_h(F) - \tilde{E}_N)^m P_N \|$
as $N^{-1}$, and we obtain
\bea\label{eq:tracenorm1}
\| (\Pi_N^0)^\perp (S_h(F) - \tilde{E}_N)^m P_N \|_1 &\leq & \| (\Pi_N^0)^\perp P_N \| \| P_N \|_1 \|(S_h(F) - \tilde{E}_N)^m P_N \| \nonumber \\
  & \leq & d_N N^{- \alpha - 1} .
\eea
For the second term on the right in (\ref{eq:decomp2}), we have
\bea\label{tracenorm2}
\| (\epsilon(h) h^2 )^m \Pi_N^0 R_{m,N} \|_1 & \leq & (\epsilon(h) h^2 )^m  \| \Pi_N^0 \|_1 \| R_{m,N} \| \nonumber \\
 &\leq  & d_N (\epsilon(h) h^2 )^{m} \| R_{m,N}\| .
\eea
The third term is estimated as
\bea\label{eq:tracenorm3}
\| (\Pi_N^0 W_h(F) \Pi_N^0 )^m P_N^\perp \|_1 &\leq & \| \Pi_N^0 \|_1 \|(\Pi_N^0 W_h(F) \Pi_N^0 )^m \| \| \Pi_N^0 P_N^\perp \| \nonumber \\
 &\leq & d_N (\epsilon (h) h^2)^m N^{- \alpha} .
\eea

Finally, we write
\bea\label{eq:tracenorm4}
Tr ((S_h(F) - \tilde{E}_N)^m P_N ) &=& Tr  (\Pi_N^0 W_h(F) \Pi_N^0 )^m  \nonumber \\
 & & + Tr \{ (S_h(F) - \tilde{E}_N)^m P_N - (\Pi_N^0 W_h(F) \Pi_N^0 )^m \} \nonumber \\
\eea
with
\bea\label{tracenorm5}
\lefteqn{ | Tr \{ (S_h(F) - \tilde{E}_N)^m P_N - (\Pi_N^0 W_h(F) \Pi_N^0 )^m \} | } \nonumber \\
 & \leq & \| (S_h(F) - \tilde{E}_N)^m P_N - (\Pi_N^0 W_h(F) \Pi_N^0 )^m \|_1
\nonumber \\
 & \leq &  d_N (\epsilon (h) h^2)^m N^{- \alpha} .
\eea
Finally, restoring the complex parameter $\theta$, we analyze the function
\beq\label{eq:analycont2}
\xi (\theta) \equiv Tr  ( (\Pi_N^0 (\theta) W_h(F, \theta) \Pi_N^0 (\theta) )^m ).
\eeq
For $\theta \in \R$, we have,
\bea\label{eq:analycont3}
\xi (\theta) & = & Tr  (D_{exp ( \theta)} (\Pi_N^0  W_h(F) \Pi_N^0 )^m D_{exp (-\theta)} ) \nonumber \\
 &=& Tr ( (\Pi_N^0 W_h(F) \Pi_N^0  )^m ) .
\eea
The function $\theta \rightarrow \xi (\theta)$ is analytic in a neighborhood of the real axis and
independent of $\theta$ on the real axis, and is thus constant. Hence we can write
\beq\label{eq:tracenorm6}
Tr ((S_h(F, \theta) - \tilde{E}_N)^m P_N (\theta) ) = Tr ( (\Pi_N^0 W_h(F) \Pi_N^0 )^m ) + E_{m,N}(\theta) ,
\eeq
where $E_{m,N}(\theta) = \mathcal{O} ((\epsilon (h) h^2 )^m N^{-\alpha})$. This completes the proof of Theorem \ref{thm:trace1}.


%
%
%
%

\section{Trace estimate for the Stark perturbation of the hydrogen atom}\label{sec:poly1}

The next step in the proof of Theorem \ref{thm:main1} consists of evaluating the trace on the right side
of (\ref{eq:trace2}). Let $\tilde{W}_h (F) = h^2 F x_1$.
The sum on the right side of (\ref{eq:trace2}) is
\beq\label{eq:trace3}
\frac{1}{d_N} \sum_{i=1}^{d_N} \left( \frac{\tau_{N,i}}{h^2 \epsilon (h)} \right)^m
= \frac{1}{d_N} Tr ((\Pi_N^0 \tilde{W}_h(f) \Pi_N^0)^m ).
\eeq
note that $\tilde{W}_h (F) = W_h (F) / (h^2 \epsilon (h)).$
Since $\tilde{W}_h (F)$ is a polynomially-bounded perturbation, we use a
general result of Thomas-Villegas-Blas \cite[Theorem 4.2]{thomas-villegas1}
in order to evaluate the semiclassical limit of the
expression on the right in (\ref{eq:trace3}) as $N \rightarrow \infty$.

\subsection{Polynomially-bounded perturbations}\label{subsec:poly1}

The main result of \cite[Theorem 4.2]{thomas-villegas1}
on the semiclassical limit for polynomially bounded perturbations is the following theorem. We slightly
change notation from \cite{thomas-villegas1} and write $h = 1/N$ and
$k = N$.

\begin{theorem}\cite[Theorem 4.2]{thomas-villegas1}\label{th:szego-limit1}
Let $V$ be a polynomially bounded, continuous function on $\R^3$ and let $g: \R \rightarrow \R$
be continuous. Then, we have
\beq\label{eq:szego-limit1}
\lim_{N \rightarrow \infty} \frac{1}{N^2} Tr ( \Pi_N^0 g ( \Pi_N^0 D_{N^{-2}} V D_{N^2} \Pi_N^0 ))
= \int_{\alpha \in \mathcal{A}} ~ g \left( \frac{1}{2 \pi} \int_0^{2 \pi} ~V(x(t, \alpha)) ~dt \right) ~d \mu ( \alpha) .
\eeq
\end{theorem}

We remark that when $V$ is a polynomial, as in our case, the proof of Theorem \ref{th:szego-limit1} is easier.
In order to apply this result to (\ref{eq:trace3}), we take $g(s) = s^m$ and $V(x) = Fx_1$.
Then we have from (\ref{eq:dilation0}) that $( D_{N^2}^{-1} V D_{N^2} ) (x) = \tilde{W}_h(F) (x)$ with
$h = 1 / N$.

Using the decomposition of the projector $\Pi_N^0$ into coherent states \eqref{eq:proj-coherentstates1},
the trace on the left in (\ref{eq:szego-limit1}), with $g(s) = s^m$,
may be expressed as an $m$-fold multiple integral with respect to $\alpha_j \in \mathcal{A}$
of matrix elements, in the momentum representation, given by
\beq\label{eq:matrixele22}
(i / N) \langle J^{1/2} K \Phi_{\alpha_i ,N} , (F \nabla_{p_1} ) J^{1/2} K \Phi_{\alpha_{i+1}, N} \rangle,
\eeq
where $\Phi_{\alpha, N}$ is the function of $\Sp^3$ defined in \eqref{eq:cost-defn1}.
Lemma 4.3 of \cite{thomas-villegas1} states that for any $\delta > 0$, the matrix element in \eqref{eq:matrixele22}
is given by
\bea\label{eq:szego-limit2}
(i / N) \langle J^{1/2} K \Phi_{\alpha_i,N} ,  J^{1/2} K \Phi_{\alpha_{i+1}, N} \rangle
 \left( \frac{1}{2 \pi}
\int_0^{2 \pi} ~F ~x(t, \alpha_{i})_1 ~dt  + \mathcal{O}( N^{\delta - 1/2} ) \right) &   \nonumber \\
   +  \mathcal{O}( N^{- \infty} ), &  \nonumber \\
\eea
where $x(t, \beta)_1$ is the first component of the vector $x(t, \beta) \in \R^3$. This vector and a corresponding momentum vector
$p(t, \beta)$, form a solution to Hamilton's equations for motion for the Hamiltonian $h(x, p) = (1/2) p^2 -  |x |^{-1}$ with energy $-1/2$.
The parameter $\beta \in \mathcal{A}$ labels the Kepler orbit with energy $-1/2$.
In the limit as $N \rightarrow \infty$, the matrix elements approach zero unless $\alpha_i = \alpha_{i-1}$.
This reduces the multiple $m$-fold integral to a single integral over $\mathcal{A}$ of the $m^{\rm th}$-power
of the integral
in \eqref{eq:szego-limit2}.


\subsection{Conclusion of the proof of Theorem \ref{thm:main1}}\label{subsec:proof-main1}

We have now proved the following result. For $\epsilon (h) = h^{6 + \delta}$, with $0 < \delta < 1$
and $h = 1/N$, the resonance shifts $z_{N, i}(F, 1/N)$ satisfy, for any $m \in \N$, the limit
\beq\label{eq:first1}
\lim_{N \rightarrow \infty} \sum_{j=1}^{d_N} \left( \frac{ z_{N,i}(F, 1/N) - E_N (1/N)}{ \epsilon(1/N)}
\right)^m
= \int_{\alpha \in \mathcal{A}} ~ \left( \frac{1}{2 \pi} \int_0^{2 \pi} ~ [ F x(t, \alpha)_1 ] ~dt \right)^m ~d \mu ( \alpha) .
\eeq
To finish the proof of Theorem \ref{thm:main1}, two steps remain to be done.
First, we re-write the integral over $\mathcal{A}$
on the right in \eqref{eq:first1} in terms of the integral over the energy surface $\Sigma (- 1/2)$.
Second, we show how to replace the monomial $s^m$ by a function
$\rho$ that is analytic in a fixed disk about the
origin.

As for the first task, we refer to \cite[pages 141-142]{uribe-villegas1}. It is noted there that the push-forward of the measure $\mu$
on $\mathcal{A}$ is the Liouville measure $\mu_L$ on the energy  surface
$\Sigma (-1/2)$. Furthermore, the Kepler orbit corresponds to the Kepler flow $\tilde{\phi}_t(x,p)$ on this energy surface.

As for the second task, the function $\rho$ of Theorem \ref{thm:main1}
is analytic in a disk of radius $3F$ about the origin. Since the perturbation $F \cdot (\tilde{\phi}_t (x,p) )_1$
is bounded by $2F$ for orbits $\tilde{\phi}_t$ on the energy surface $\Sigma (-1/2)$,
the estimate \eqref{eq:first1} guarantees convergence in the trace of the power series expansion for $\rho$.
This concludes the proof of Theorem \ref{thm:main1}.


\section{Appendix 1: Coherent states for the hydrogen atom}\label{sec:app1-coherentstates1}

We review the construction and properties of the coherent states that form an overcomplete set in the eigenspace $\mathcal{E}_\ell$ of
the hydrogen atom Hamiltonian corresponding to the eigenvalue $E_\ell = - 1/ (2 h^2 \ell^2)$.
Let $\mathcal{A}$ be the five real-dimensional subspace of $\C^4$
defined by
\beq\label{eq:defnA1}
\mathcal{A} = \{ \alpha = (
\alpha_1, \alpha_2, \alpha_3, \alpha_4 ) ~|~ \alpha_j \in \C, ~\|
\Re \alpha_j \| = \| \Im \alpha_j \| = 1, ~ \Re \alpha \cdot \Im
\alpha = 0 \} .
\eeq
This provides a parametrization of the co-sphere
bundle $S^* \Sp^3$ of the three-sphere. There is a
$SO(4)$-rotationally invariant probability measure on $\mathcal{A}$
that we denote by $\mu$. Coherent states on $\Sp^3$ have the form
\beq\label{eq:cost-defn1}
\Phi_{\alpha, \ell}(\omega) = a(\ell) ( \alpha
\cdot \omega)^\ell, ~~\omega \in \Sp^3, ~~ \alpha \in \mathcal{A}, ~~\ell \in 0, 1, 2, \ldots.
\eeq
The coefficient $a (\ell) \sim \ell^{1/2}$ is fixed by the requirement
that the $L^2 ( \Sp^3)$-norm of $\Phi_{\alpha, N}$ is equal to one, see \cite[(2.11)]{thomas-villegas1}.
These states are hyper-spherical harmonics. They
are eigenstates of the spherical Laplacian $- \Delta_{\Sp^3}$
with eigenvalue $\ell (\ell + 2)$. The entire family $\{ \Phi_{\alpha, \ell}(\omega) ~|~
\alpha \in \mathcal{A} \}$ is over complete and spans the eigenspace
of $- \Delta_{\Sp^3}$ with eigenvalue $\ell(\ell+2)$.
We note that these states have the property that as $N \rightarrow \infty$ they concentrate on the
great circle $\{ \omega \in \Sp^3 ~|~ | \alpha \cdot \omega | = 1 \}$
generated by the real and imaginary parts of $\alpha$.

In momentum space $\R^3$, the coherent states have the
following form. The inverse of the stereographic projection from the three
sphere $\Sp^3$ to $\R^3$ is the mapping $p \in \R^3
\rightarrow \omega (p) \in \Sp^3$ defined by
\bea\label{eq:stereo1}
\omega_j (p) &=& \frac{2 p_j}{ \|p\|^2 + 1}, ~j=1,2, 3 \nonumber \\
\omega_4 (p) &=& \frac{\|p\|^2 - 1 }{\| p \|^2 + 1} .
\eea
For any $\alpha \in \mathcal{A}$, we define
\beq\label{eq:cohst-p1}
\hat{\Psi}_{\alpha, \ell} (p) = a(\ell-1) \ell^{3/2} \left( \frac{2}{\ell^2 \|p\|^2 +1} \right)^2 (\alpha \cdot \omega (\ell p) )^{\ell-1},
 ~p  \in \R^3 .
\eeq
These functions are in $L^2 ( \R^3)$. Their Fourier transforms are eigenfunctions of $H_V$ with
eigenvalue $E_\ell$ \cite[section II.B]{BI:1966}.
They form an overdetermined basis of the $\ell^2$-dimensional eigenspace $\mathcal{E}_\ell \subset L^2 (\R^3)$
in the momentum space representation.

For $\theta \in \R$, we scale these momentum space functions by $e^\theta$ to obtain
\beq\label{eq:coh-analy1}
\hat{\Psi}_{\alpha, \ell} (e^\theta p) = a(\ell-1) \ell^{3/2}
\left( \frac{2}{\ell^2 e^{ 2 \theta} \|p\|^2 +1} \right)^2 (\alpha \cdot \omega (\ell e^\theta p) )^{\ell-1}, ~p  \in \R^3 .
\eeq
These functions are analytic $L^2$-valued functions for $| \Im \theta | < \pi / 2$.

The configuration space coherent states are obtained by the inverse Fourier transform:
\bea\label{eq:cohst-p2}
{\Psi}_{\alpha, \ell} (x) &=& \frac{1}{(2 \pi)^{3/2}} \int_{\R^3} ~d^3 p ~e^{i x \cdot p} ~ \hat{\Psi}_{\alpha, \ell} (p) \nonumber \\
 &=& \frac{a(\ell-1) \ell^{3/2}}{ (2 \pi)^{3/2}} \int_{\R^3} ~d^3 p ~e^{i x \cdot p}
~\left( \frac{2}{\ell^2 \|p\|^2 +1} \right)^2 (\alpha \cdot \omega (\ell p) )^{\ell-1} . \nonumber \\
\eea
These form an overdetermined basis of normalized (but not orthogonal)
$L^2$-functions for $\mathcal{E}_\ell \subset L^2 (\R^3)$ in the configuration
space picture. Let $\mu$ be the probability measure on $\mathcal{A}$.
The orthogonal projector $\Pi_\ell^0$ onto the eigenspace $\mathcal{E}_\ell$
may be written as
\beq\label{eq:proj-coherentstates1}
\Pi_\ell^0  =  \ell^2 \int_{\mathcal{A}} ~| \Psi_{\alpha, \ell} \rangle \langle \Psi_{\alpha, \ell}| ~d \mu(\alpha) .
\eeq

For $\theta \in \R$, the dilated coherent states are
\bea\label{eq:coherent-dila1}
\Psi_{\alpha, \ell} (x; \theta) & \equiv & D_{exp (\theta)} \Psi_{\alpha, \ell} (x) = e^{3 \theta / 2} \Psi_{\alpha, \ell} (e^\theta x) \nonumber \\
 &=& \frac{e^{-3 \theta/ 2}}{ (2 \pi )^{3/2}} \int_{\R^3} ~d^3 p ~e^{i x \cdot p} ~ \hat{\Psi}_{\alpha, \ell} (e^{- \theta} p) \nonumber \\
 & = & a(\ell - 1) \ell^{3/2} \frac{e^{-3 \theta/2}}{ (2 \pi )^{3/2}} \int_{\R^3} ~d^3 p ~e^{i x \cdot p}
 \left( \frac{2}{\ell^2 e^{-2 \theta} \|p\|^2 +1} \right)^2 (\alpha \cdot \omega (\ell e^{-\theta} p) )^{\ell-1}  . \nonumber \\
\eea
From the comments after \eqref{eq:coh-analy1}, it follows that $\Psi_{\alpha, \ell} (x; \theta)$ has an
analytic continuation as an $L^2$-valued function of $\theta$ for $| \Im \theta | < \pi / 2$.
We also consider the dilated projector $\Pi_\ell^0 (\theta) = D_{exp(\theta)} \Pi_\ell^0 D_{exp (- \theta)}$.
We extend these dilated operators to $\theta \in \C$ provided $|\Im \theta | < \pi / 2$.
As mentioned above, the coherent states concentrate on the Kepler orbit as $\ell \rightarrow \infty$.
Lemma \ref{lemma:moments1} is a consequence of this fact.






\section{Appendix 2: Moments of the hydrogen atom eigenfunctions}\label{sec:app2-ef}

We give the brief proof of the localization property of the eigenstates of the hydrogen atom Hamiltonian
discussed in section \ref{subsec:dlated-cohst1}.
The eigenfunctions of $H_V$ may be written as
$\psi_{n \ell m}(r, {\tilde{\theta}}, \phi) = e^{- r /n }F_{nl}(r) Y_{lm}( {\tilde{\theta}}, \phi)$,
with $n  = 1, 2, \ldots$, the principal quantum number. The functions
$Y_{\ell m}$ are the spherical harmonics $(\ell, m)$ labeling the angular momentum
so that $0  \leq \ell \leq n-1$ and $- \ell \leq m \leq \ell$.
The radial component $F_{nl} (r) = A_{n \ell} (2 r / n)^{\ell} L_{n-\ell-1}^{2\ell+1}(2 r/ n)$,
where $L_{n-\ell-1}^{2\ell+1}(r)$ is the associated Laguerre polynomial of degree $n- \ell - 1$.
The normalization constant is $A_{n \ell} = (2 / n^2) \sqrt{ (n - \ell -1)! / [ (n + \ell )! ]^3 }$.
The dilated eigenfunctions $\psi_{n \ell m} (\theta) (r) \equiv
(D_{exp( \theta)} \psi_{n \ell m} )(r) = e^{ 3 \theta / 2} \psi_{n
\ell m} (e^{\theta} r)$ are analytic $L^2$-valued functions in the
strip $|\Im \theta | < \pi / 2$. Consequently, the function
$f(\theta) \equiv \langle \psi_{n \ell m}(\theta) , \| x \|^k
\psi_{n \ell' m'}(\theta) \rangle$ is analytic on the strip $| \Im
\theta | < \pi / 2$. For $\theta \in \R$, we have $f(\theta) =
e^{-k \theta} f(0)$. The function $\tilde{f}(\theta) = e^{-k
\theta} f(0)$ is an entire analytic function. By the identity
principle for analytic functions, we have $\tilde{f}(\theta) =
f(\theta)$ on the strip $| \Im \theta|< \pi / 2$.


We first want to estimate the large $n$ behavior of the
expectation of powers of the position operator in the dilated
eigenstates $\psi_{n \ell m}$ as in \eqref{eq:matrixele0}. Since
$f(\theta) = e^{-m \theta} f(0)$, this reduces to an estimate on
moments of the position operator. For any $k > -2 \ell - 1$, these
moments are well-known and computed in section 3, Appendix B of
Messiah \cite{Messiah}. Let \beq\label{eq:matrixele2} \langle \| x
\|^k \rangle_{n \ell m} \equiv \langle \psi_{n \ell m}(\theta) ,
\| x \|^k
\psi_{n \ell' m'}(\theta) \rangle.
\eeq
Note that because of the spherical symmetry, the quantum numbers $(\ell, m)$ must be the same in \eqref{eq:matrixele2}.
Properties of the Laguerre polynomials leads to
\beq\label{eq:firstmoment1}
\langle \| x \| \rangle_{n \ell m}  = \frac{1}{2} [ 3n^2 - \ell (\ell+1) ],
\eeq
and a recursion formula for higher moments $k \geq 2$:
 \beq\label{eq:k-moment1}
\langle \| x \|^k \rangle_{n \ell m}  = n^2 \frac{2 k +1}{k + 1}
\langle \| x \|^{k-1} \rangle_{n \ell m} - \frac{n^2 k}{4(k+1)} [
(2 \ell +1)^2  - k^2 ] \langle \| x \|^{k-2} \rangle_{n \ell m}.
\eeq It is easy to verify from \eqref{eq:firstmoment1} and
\eqref{eq:k-moment1} that \beq\label{eq:matrixele222} \langle
\psi_{n \ell m} , \| x \|^k
\psi_{n \ell m}(\theta) \rangle = \mathcal{O}( n^{2k} ),
\eeq
since $0 \leq l \leq n-1$ (see \cite[section B.3]{Messiah}), verifying \eqref{eq:matrixele0}.

Secondly, we want to compute the matrix elements of the Stark perturbation.
It is easier to do the calculation with the electric field in the three direction $x_3$
because of the properties of the usual spherical harmonics.
The result is independent of the field direction so we assume this here.
Recalling the identity principle for analytic functions used above,
the matrix element of the Stark perturbation is
\bea\label{eq:matrixele1}
\langle \psi_{n \ell m}(\theta) , W_h(F, \theta) \psi_{n \ell' m'}(\theta) \rangle &= & \epsilon (h) h^2 e^\theta F
\langle \psi_{n \ell m}(\theta) , x_3 \psi_{n \ell' m'}(\theta) \rangle \nonumber \\
 &=& \epsilon (h) h^2 F \langle \psi_{n \ell m}, x_3 \psi_{n \ell' m'} \rangle , ~~ | \Im \theta | < \pi / 2  .
 \nonumber \\
 & &
\eea
We write $x_3= r \cos \tilde{\theta}$, where $\tilde{\theta}$ is the angle with the $x_3$-axis.
By the Cauchy-Schwarz inequality and the first moment estimate \eqref{eq:firstmoment1}, we have
\beq\label{eq:stark-moment1}
| \langle \psi_{n \ell m}, x_3 \psi_{n \ell' m'} \rangle |  \leq \| r^{1/2} \psi_{n \ell' m'} \|
~\| r^{1/2} \psi_{n \ell m} \|  = \mathcal{O}( n^2) ,
\eeq
since $0 \leq l \leq n-1$ (see \cite[section B.3]{Messiah}.)


\section{Appendix 3: The quadratic estimate}\label{appendix:QE}

We restate the quadratic estimate appearing in Proposition II.4 of \cite{herbst1}.
In the notation of that paper, let $h(\alpha) = - \Delta + \alpha x_1$ where $| \Im \alpha | > 0$.
Let $\theta = \arg \alpha$ and define two constants
\beq\label{eq:constants1}
c(\alpha) = (3/2)(1 - | \cos \theta |)^{3/2} | \sin \theta |^{3/2} | \alpha|^{4/3}, ~~
\beta (\alpha) = (1 - | \cos \theta |) / 2.
\eeq
then, for all $\psi \in S (R^3)$, we have
\beq\label{eq:QE2}
\| h(\alpha) \psi \|^2 + c(\alpha) \| \psi \|^2 \geq \beta (\alpha) ( \| \Delta \psi \|^2
+ \| x_1 \psi \|^2 ).
\eeq
In our case, we have
\beq\label{eq:QE3}
H_0 (h, F, \theta) = e^{-2 i \theta} [ - \Delta + e^{3 i \theta} h^4 \epsilon (h) F x_1 ]
=e^{-2 i \theta} h(\tilde{\alpha}),
\eeq
where $\tilde {\alpha} = e^{3 i \theta} h^4 \epsilon (h) F $.
From \eqref{eq:constants1}, we find that
\beq\label{eq:constants2}
c(\tilde{\alpha}) = \mathcal{O}(N^{- (4/3)(4 + K + \delta)}), ~~ \beta(\tilde{\alpha})
= \mathcal{O}(1).
\eeq
Consequently, we have the following bound for $z \in \gamma_N$:
\beq\label{eq:QE4}
\| x_1 ( H_0(h,F,\theta) - z)^{-1} \| = \mathcal{O}(1).
\eeq


\end{document}